\documentclass[11pt]{article}
\RequirePackage[T1]{fontenc}
\usepackage[english]{babel}
\RequirePackage{amsthm,amsmath}
\usepackage[a4paper]{geometry}
\usepackage{setspace}
\usepackage[usenames, dvipsnames]{color}
\usepackage{graphicx}
\usepackage{bm}
\usepackage{mathrsfs}
\usepackage[font={small,it}]{caption}
\usepackage{bbold}
\usepackage{multirow}
\usepackage{bbm}
\usepackage{chngpage}
\usepackage{adjustbox}
\usepackage{nicefrac}
\usepackage{amssymb}
\usepackage{bigints}
\usepackage{algorithm}
\usepackage{algpseudocode}
\usepackage{natbib}
\RequirePackage[colorlinks,citecolor=blue,urlcolor=blue]{hyperref}
\usepackage{xr}

\newtheorem{theorem}{Theorem}[section]

\newtheorem{remark}[theorem]{Remark}
\newtheorem{definition}[theorem]{Definition}

\newcommand\red[1]{{\color{black}#1}}
\begin{document}
\begin{center}

\textbf{\Large Exhaustive goodness-of-fit \\via smoothed inference and graphics}
\vspace{1cm}

Sara Algeri$^{1,\dag}$ and Xiangyu Zhang$^1$\\
\vspace{0.5cm}
{\small $^1$ School of Statistics, University of Minnesota, \\0461 Church St SE, Minneapolis, MN 55455, USA.}\\
{\small $^\dag$ Email: salgeri@umn.edu}\\
\end{center}

\begin{abstract}
Classical tests of goodness-of-fit aim to validate the conformity of a postulated model to the data under study. Given their inferential nature, they can be considered a crucial step in confirmatory data analysis. In their standard formulation, however, they do not allow exploring how the hypothesized model deviates from the truth nor do they provide any insight into how the rejected model could be improved to  better fit  the data.
The main goal of this work is to establish a comprehensive framework for goodness-of-fit which naturally integrates modeling, estimation, inference and graphics. Modeling and estimation focus on a novel formulation of  smooth tests that easily extends to arbitrary distributions, either continuous or discrete. Inference and adequate post-selection adjustments are performed via a specially designed smoothed bootstrap and the results are summarized via an exhaustive graphical tool called \emph{CD-plot}. \red{Technical proofs, codes and data are provided in the Supplementary Material. }
\end{abstract}

{\small
Goodness-of-fit, smooth tests,  smoothed bootstrap, graphical inference.
}


\section{Introduction}
Tests for goodness-of-fit such as \cite{pearson}, \cite{anderson}
and \cite{shapiro} are some of the most popular methods used to assess \underline{if} a model postulated by the scientists deviates significantly from the true data distribution. Because of their inferential nature, they can be framed in the context of confirmatory data analysis but they provide little or no insight from an exploratory and/or modeling perspective.
Specifically, when the postulated model is rejected, they do not equip practitioners with any insight on \underline{how} the latter deviates from the truth nor do they provide indications on how the rejected model can be improved to   better fit  the data.

A more comprehensive approach to goodness-of-fit that naturally addresses these drawbacks is given by smooth tests. They were first introduced by \citet{neyman37} as a generalization of the Pearson $\chi^2$ test and were further extended by \citet{barton53, barton55,barton56}. The main idea is to conduct a test of hypothesis where the alternative model embeds the null as a special case through a series of orthonormal functions. As a result, if the hypothesized model is rejected, the alternative model naturally allows correcting it to provide a better fit for the data.
Despite the existence of smooth tests for regular distributions \citep[e.g.,][]{rayner86, rayner88}, the specification of the orthonormal functions being used is not universal and typically depends on the distribution under comparison. Furthermore, the resulting inference is strongly affected by the model selection process and often relies on asymptotic results.

The main goal of this work is to generalize smooth tests by introducing a unifying framework which (i) easily generalizes to both continuous and discrete data in arbitrarily large samples, (ii) naturally leads to an efficient sampling scheme to perform variance estimation and inference while accounting for model selection and (iii) allows us to graphically assess  where a significant deviation of the hypothesized model from the truth occurs.
In our strategy, the key ingredient to tackle (i)  is  a specially designed orthonormal basis first introduced by 
\citet{LP13} and \citet{LP14} in the context of the so-called \emph{LP approach to statistical modeling}\footnote{In the \emph{LP} acronym, the letter \emph{L} typically denotes nonparametric methods based on quantiles, whereas \emph{P} stands for polynomials \cite[Supp S1]{LPksamples}.}. To address (ii) we study and further extend a novel smoothed bootstrap which naturally applies to both continuous and discrete data and, most importantly, allows us to efficiently perform variance estimation, inference and adequate post-selection adjustments  for arbitrarily large samples. To achieve (iii) we propose the so-called \emph{Comparison Density plot} or \emph{CD-plot}, a graphical tool for goodness-of-fit that allows us to assess simultaneously \underline{if} and \underline{how} significant deviations of the hypothesized from the true underlying model occur. 
The \emph{CD-plot}  was originally proposed in \citet{algeri20} in the context of signal detection of astrophysical searches involving large continuous-valued data samples. Here, we extend \citet{algeri20} to arbitrary large samples from either continuous or discrete distributions. 

Furthermore, we introduce a novel sampling scheme, called the \emph{bidirectional acceptance sampling algorithm}, which allows us to simulate from two different distributions simultaneously and, consequently, improves the computational performance of the LP-smoothed bootstrap in our setting.
Finally, we apply the methods proposed to identify the distribution of the time from the onset of symptoms to hospitalization of COVID-19 patients. 

The remainder of the article is organized as follows. In Section \ref{overview}, we  review smooth tests and in Section \ref{framework} we reformulate them in the context of LP modeling. In Section \ref{LPboot} we introduce the CD-plot and the LP-smoothed bootstrap. Section \ref{double} is dedicated to bidirectional acceptance sampling. Important extensions of the method proposed are covered in Section \ref{extensions}. The analysis of COVID-19 hospitalization data is presented in Section \ref{COVID}. A discussion is proposed in Section \ref{discussion}. \red{ Technical proofs, data and the \texttt{R} codes and data needed to reproduce the results presented in this paper are provided among the Supplementary Material.
Finally, the \texttt{R} package \texttt{LPsmooth} \citep{LPsmooth} aims to facilitate the implementation of the methods proposed in practical applications using \texttt{R} programming \citep{Rprogramming}. }

\section{Background: comparison distributions and smooth tests}
\label{overview}
Smooth tests are a class of goodness-of-fit inferential methods that rely on the specification of a \emph{smooth model}, of which the model hypothesized by the researcher is a special case. Several classes of smooth models have been proposed in the literature \citep[e.g.,][]{neyman37,barton56} and involve an orthonormal expansion of the so-called \emph{comparison density} \citep{parzen83}.
\begin{definition}
Let $X$ be a random variable, either discrete or continuous, with probability function $f$ and cumulative distribution function (cdf) $F$ and let $g$ be a suitable probability function, with cdf $G$, same support of $F$, and quantile function $G^{-1}$. The comparison density between $F$ and $G$ can then be specified as
\begin{equation}
\label{cd}
d(u;G,F)=\frac{f\bigl(G^{-1}(u)\bigl)}{g\bigl(G^{-1}(u)\bigl)} \qquad\text{with $u=G(x)$,}
\end{equation}
and we assume that $f=0$ whenever $g=0$.
\end{definition}
In \eqref{cd}, $g$ is known in the literature as \emph{parametric start} \citep{hjort95} or \emph{reference distribution} \citep{morris}.
The \emph{comparison distribution} is defined as $D(u)=\int_0^ud(s;G,F)\text{d} s$. As noted by \citet{parzen2004}, in the continuous case $D(u)=F(G^{-1}(u))$ for all $u\in[0,1]$; whereas, for $X$ discrete, $D(u)$ is piecewise linear at values $u_r=G(x_r)$, where $x_1<\dots< x_R$ are probability mass points of $X$ and\begin{equation}
\label{relationship}
D(u_r)=F(G^{-1}(u_r))=F(x_r).\end{equation}
It follows that, in the discrete case, the comparison density is a step function (e.g., bottom left panel in Figure \ref{fig1}) with values $d(u; G,F)=f(x_r)/g(x_r)$ for $G(x_{r-1})<u\leq G(x_r)$ \citep[see][p.18, for more details]{morris}.
Throughout the manuscript, we will mainly consider the case where the parametric start is fully known; the case where $g$ is characterized by a set of unknown parameters is discussed in Section \ref{extensions}.

On the basis of \eqref{cd}, a smooth model is
\begin{equation}
\label{skewG}
f_m(x)=g(x)d_m(G(x);G,F),
\end{equation}
where $d_m(G(x);G,F)$ is a representation of \red{the comparison density in} \eqref{cd} by means of a series of $m$ functions of $G(x)$, denoted by $h_j(G(x))$, which together form a complete orthonormal basis on $[0,1]$. Table \ref{families} summarizes possible specifications of $d_m(G(x);G,F)$ proposed in the literature. Clearly, $d_m(x)=d(x)$ and $f_m(x)=f(x)$, whenever $m=R-1$, if $X$ is discrete, with $R$ being the total number of distinct mass points of $X$, or when $m=\infty$, if $X$ is continuous. For the moment, we consider $m$ to be chosen arbitrarily; a discussion on the choice of $m$ is postponed to Section \ref{extensions}.

\begin{table}
{\fontsize{3.3mm}{3.3mm}\selectfont{
\begin{center}
\begin{tabular}{|c|c|}
\hline
\vspace{-0.3cm}
& \\
Method &$d_m(u;G,F)$\\
\vspace{-0.3cm}
& \\
\hline
\vspace{-0.2cm}
& \\
Neyman (1937) & $\exp\bigl\{\tau_0+\sum_{j=1}^{m}\tau_jh_{j}(u)-K_{\bm{\tau}}\bigl\} $ \\
\vspace{-0.2cm}
& \\
Barton (1953) & $1+\sum_{j=1}^{m}\theta_jh_{j}(u) $ \\
\vspace{-0.2cm}
& \\
Devroye-Gy\"orfi (1985)& $\max\bigl\{0,1+\sum_{j=1}^{m}\theta_jh_{j}(u)\bigl\}/K_{\bm{\theta}}$\\
\vspace{-0.2cm}
& \\
Gajek (1986) & $\max\bigl\{0,1+\sum_{j=1}^{m}\theta_jh_{j}(u)-K_{\bm{\theta}}\bigl\}$\\
\hline
\end{tabular}
\end{center}}}
\caption{Possible representations of $d_m(u;G,F)$. The functions $h_{j}(u)$ form a complete orthonormal basis on
$[0,1]$, and $K_{\bm{\tau}}$, $K_{\bm{\theta}}$ are normalizing constants; whereas, $\bm{\tau}=(\tau_0,\dots,\tau_m)$ and $\bm{\theta}=(\theta_1,\dots,\theta_m)$ denote vectors of unknown coefficients such that
$\theta_j=\int_0^1h_j(u)d_m(u;G,F)$, for all $j=1,\dots,m$. }
\label{families}
\end{table}
Finally, given a set of independent and identically distributed (i.i.d.) observations $x_1,\dots,x_n$ from $X_1,\dots,X_n$, with $X_i\sim F$, for all $i=1,\dots,n$, a smooth test is constructed by testing, for any of the models in Table \ref{families}, the hypotheses $H_0:\bm{\theta}=0$ versus $H_1:\bm{\theta}\neq 0$ (or $H_0:\bm{\tau}=0$ versus $H_1:\bm{\tau}\neq 0$ for \citet{neyman37}). The test statistic takes the form\begin{equation}
\label{v}
W=\sum_{j=1}^m\biggl[\frac{1}{\sqrt{n}}\sum_{i=1}^nh_j(G(x_i))\biggl]^2,
\end{equation} and follows, asymptotically, a $\chi^2_m$ distribution under $H_0$ (see \citet[Ch. 4]{thas} for a self-contained review on smooth tests). Finally, a rejection of $H_0$ implies that the smooth model fits the data significantly better than the hypothesized model $g$.

\section{Smooth tests via LP score functions}
\label{framework}
The smooth tests discussed in Section \ref{overview} can be applied to both continuous and discrete distributions \citep[e.g.,][Ch. 8]{rayner2009}, however, they often require the specification of an adequate orthonormal system on a case-by-case basis \citep[e.g.,][Ch. 9-11]{rayner2009}. The goal of this section is to exploit the novel LP approach to statistical modeling first introduced by \citet{LP14} to provide a generalized formulation of smooth tests that extends to arbitrary distributions.

\red{
\subsection{The LP score functions}
The LP-score functions first  introduced by \citet{LP14}, can be constructed as follows. } Given a random variable $X$, either continuous or discrete, we can construct a complete orthonormal basis of LP score functions for $G(x)$ by setting the first component to be $T_0(x;G)=1$, whereas subsequent components $\{T_j(x;G)\}_{j>0}$ can be obtained by Gram--Schimidt orthonormalization of powers of
\begin{equation}
\label{T1}
T_1(x;G)=\frac{G^{\text{mid}}(x)-0.5}{\sqrt{[1-\sum_{x\in \mathcal{X}}g^3(x)]/12}},
\end{equation}
where $\mathcal{X}$ is the set of distinct points in the support of $X$, $g(x)=P(X=x)$ when $X\sim G$ and $G^{\text{mid}}(x)=G(x)-0.5p_G(x)$ is the so-called \emph{mid-distribution function}, with mean and variance given by $0.5$ and $[1-\sum_{x\in \mathcal{U}}p^3_G(x)]/12$, respectively \citep{parzen2004}. Interestingly, when $X$ is continuous, $G^{\text{mid}}(x)=G(x)$ and the LP score functions can be expressed as normalized shifted Legendre polynomials.

\begin{figure*}[htb]
\begin{tabular*}{\textwidth}{@{\extracolsep{\fill}}@{}c@{}c@{}}
\hspace{-1cm}\includegraphics[width=83mm]{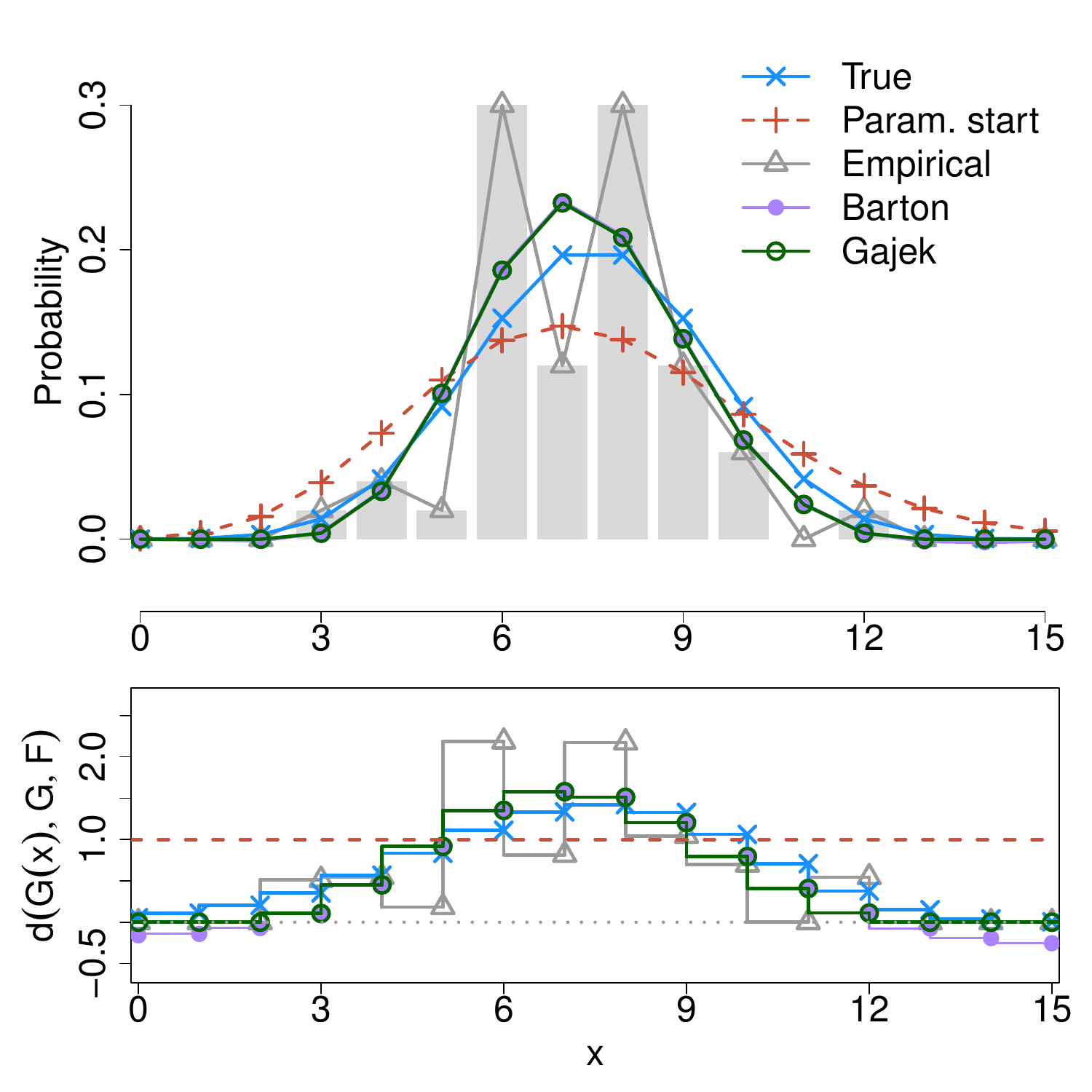} &\hspace{-0.5cm} \includegraphics[width=83mm]{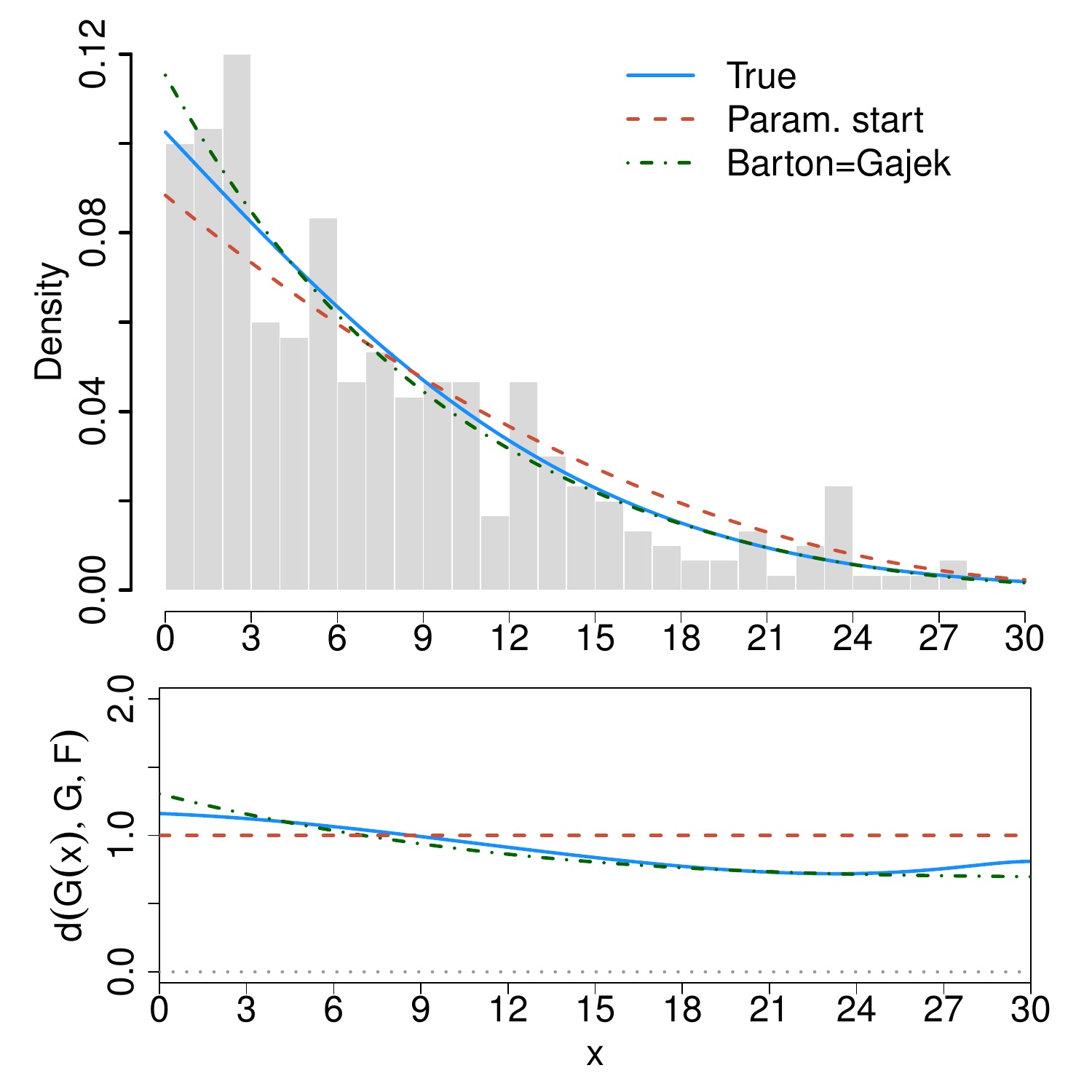}\\
\end{tabular*}
\caption{Comparison of different estimators of $f$ and $d$ when $X$ is discrete (left panels) and when $X$ is continuous (right panels). The upper left panel shows various estimators of the probability mass function (pmf) $f$ when $X\sim \text{Binomial(15,0.5)}$, $n=50$ and the parametric start is the pmf of a $\text{Poisson}(7.5)$ random variable truncated over the range $[0,15]$. The bottom left panel compares the respective comparison density estimators. In the upper right panel, the true probability density function (pdf) $f$ is that of a $\text{Normal(-15,15)}$ random variable truncated over the range $[0,30]$. In this case, $n=300$ and the parametric start is the polynomial density in \eqref{poly}. In this setting, \eqref{barton} leads to a bona fide estimate and thus it coincides with \eqref{gajek}. The respective comparison densities are shown in the bottom right panel. In all the cases considered, the smoothed estimators have been computed choosing $m=2$. }
\label{fig1}
\end{figure*}

\red{
\subsection{LP-based smooth tests}}
To generalize our framework to both the continuous and discrete settings, it is particularly useful to express the quantities of interest in the quantile domain by means of the probability integral transform $U=G(X)$. Specifically, denote with $\{S_j(u;G)\}_{j\geq0}$ the basis of LP score functions expressed in the quantile domain and with respect to the probability measure $G$, i.e., $S_j(u;G)=T_j(G^{-1}(u);G)$. If $d\in L_2[0,1]$; we can rewrite the models in Table \ref{families} in terms of $LP$ score functions by setting
\begin{align}
\label{setting}
h_j(u)=S_j(u;G),\qquad \theta_j&=LP_j=\int_0^1S_j(u;G)d(u;G,F)\text{d} u,\\
\label{setting2}
\text{and ${\bm{\tau}}$ satisfies the constraints }&\quad LP_j=\int_0^1 S_j(u;G)d_{\bm{\tau}}(u;G,F)\text{d} u,
\end{align}
where $d_{\bm{\tau}}(u;G,F)$ in \eqref{setting2} is the comparison density representation in the formulation of \citet{neyman37} (see Table \ref{families}).

From \red{the decomposition in} \eqref{skewG}, it follows that smooth estimators of $f$ can be specified as
\begin{equation}
\label{skewGest}
\widehat{f}(x)=g(x)\widehat{d}(G(x);G,F),
\end{equation}
where $\widehat{d}(G(x);G,F)$ is an estimate of the comparison density obtained by replacing suitable estimators of $LP_j$ and/or $\tau_j$ in any of the models in Table \ref{families}.
A simple strategy to estimate the $LP_j$ coefficients is to consider their sample counterparts. Specifically, let $\tilde{F}$ be the empirical cdf, i.e., $\tilde{F}(x)=\frac{1}{n}\sum_{i=1}^nI_{\{X_i\leq x\}}$, where $I_{\{\cdot\}}$ is the indicator function and denote with $\tilde{d}$ the empirical comparison density, i.e., $\tilde{d}(u;G,F)=d(u;G,\tilde{F})$. The $LP_j$ coefficients can be estimated via
\begin{align}
\label{LPj}
\widehat{LP}_j&=\frac{1}{n}\sum_{i=1}^nT(x_i;G)=\int_0^1S_j(u;G)\tilde{d}(u;G,F)\text{d} u,
\end{align}
and thus we can calculate
\begin{align}
\label{ecd}
\tilde{d}(u;G,F)&=1+\sum_{j=1}^{\tilde{m}}\widehat{LP}_j S_j(u;G)\quad \text{with}\quad \tilde{m}=\min\{R-1,n-1\}.
\end{align}
From \eqref{LPj}, it follows that
\begin{equation}
\label{momentsLP}
E[\widehat{LP}_j]=LP_j\qquad\text{and}\qquad \text{cov}(\widehat{LP}_j,\widehat{LP}_k)=\frac{\sigma_{jk}}{n},
\end{equation}
with $\sigma_{jk}=\text{cov}[S_j(U;G),S_k(U;G)]$. Finally, estimates of the coefficients $\tau_j$ can be obtained numerically by replacing $\widehat{LP}_j$ to $LP_j$ in \eqref{setting2} \citep[e.g.,][]{LPmode}.

Finally, the equivalent of \red{the test statistics in} \eqref{v} can be specified by considering the so-called deviance statistics \citep{LPfdr,algeri20}
\begin{equation}
\label{deviance}
D=n\sum_{j=1}^m\widehat{LP}^2_j,
\end{equation}
which  follows, asymptotically and under $H_0$, a $\chi^2_m$ distribution.
\red{It is worth pointing out  the main steps outlined in this section date back to \citet{neyman37} and \citet{barton53}.  Whereas, the use of the LP score functions to extend the estimation framework also to the discrete case is due to \citet{LP14}. The estimators considered in \citet{LP14} (see also \citet{LPmode}) are the Neyman and Barton's estimators (see Table \ref{families}). Whereas, in this article, the focus is on the use of the LP score functions in the inferential stage, and the estimator we rely upon is the Gajek's estimator.}

\red{
\subsection{The Gajek estimator}}
\label{estgajeksec}
Although any of the models in Table \ref{families} can be used to obtain estimates of $d$ and $f$,
in this article we mainly focus on the model proposed by \citet{gajek86} and estimated via
{\fontsize{3.3mm}{3.3mm}\selectfont{
\begin{equation}
\label{gajek}
\widehat{f}_m(x)=g(x)\widehat{d}_m(G(x);G,F)\quad\text{with}\quad \widehat{d}_m(G(x);G,F)= \max{\biggl\{0,1+\sum_{j=1}^{m}\widehat{LP}_jT_j(x;G)-K \biggl\}},
\end{equation}}}
where $K$ is chosen to guarantee that $\widehat{f}_m(x)$ integrates/sums to 1. \red{The notation $\widehat{d}_m(u; G, F)$ is used to emphasize that we are estimating the comparison density between the parametric start $G$ (i.e., the distribution in the second argument) and the true model $F$ (i.e., the distribution in the third argument). Whereas, the subscript $m$ highlights that the estimator relies on the first $m$ terms of the expansion.} Notice that, $\widehat{f}_m(x)$ is obtained by adequately correcting Barton's estimate
\begin{equation}
\label{barton}
\dot{f}_m(x)=g(x)\dot{d}_m(G(x);G,F)\quad\text{with}\quad \dot{d}_m(G(x);G,F)= 1+\sum_{j=1}^{m}\widehat{LP}_jT_j(x;G)
\end{equation}
to provide a \emph{bona fide} estimate of $f$, i.e., nonnegative and with integral/sum equal to one.
Furthermore, let the Mean Integrated Squared Error (MISE) of an estimator $\widehat{d}$ of $d$ be $MISE(\widehat{d})=\int_0^1 \bigl[\widehat{d}(u;G,F)-d(u;G,F)\bigl]^2\text{d} u$;
the work of \citet{gajek86} showed that $MISE(\widehat{d}_m) \leq MISE(\dot{d}_m)$, whereas, the same result is not guaranteed when considering other bona fide
corrections in the formulation of \citet{DG} in Table \ref{families} \citep{kaluszka}.

The estimators in \eqref{ecd}, \eqref{gajek} and \eqref{barton} are compared using two illustrative examples in Figure \ref{fig1}. The left panels show the results obtained when considering $n=50$ observations from $X\sim\text{Binomial(15,0.5)}$ with parametric start $g$ (red vertical crosses) corresponding to the pmf of a $\text{Poisson}(7.5)$ random variable truncated over the range $[0,15]$. In the upper left panel the Barton's estimator of $f$ (purple dots) leads to nonnegative values and thus the respective Gajek's correction in \eqref{gajek} is computed (green circles); both \red{estimators in} \eqref{gajek} and \eqref{barton} are computed choosing $m=2$. The two smoothed estimators of $f$ show only minor differences from one another, however, they differ substantially from the empirical mass function, i.e., $\tilde{f}(x)=\frac{1}{n}\sum_{i=1}^nI(x_i=x)$ (gray triangles) and provide estimates which are closer to the truth (blue crosses).
To highlight the differences between the discrete and the continuous settings, the right panels show the results obtained when considering $n=300$ observations from a $\text{Normal(-15,15)}$ random variable truncated over the range $[0,30]$. Here, the parametric start $g$ (red dashed lines) corresponds to the polynomial density
\begin{equation}
\label{poly}
g(x)=\frac{1}{w}[4.19-0.25x+0.0038x^2]\qquad\text{with $x\in[0,30]$},
\end{equation}
and $w$ is a normalizing constant. In this case, choosing $m=2$, the Barton estimator in \eqref{barton} leads to a bona fide estimate (green dotted--dashed line) and coincides with \red{the Gajek estimator in }\eqref{gajek}.

\red{A useful feature of} estimators of the form in \eqref{skewGest} is that the graph of $\widehat{d}(u;G,F)$ allows us to visualize  where  and  how  the true model of the data deviates from the hypothesized model $g$. For our toy examples, the graphs of different estimators of $d$ are displayed in the bottom panels of Figure \ref{fig1}. In the discrete case, the comparison density estimators considered are below one in correspondence of the most extreme quantiles, suggesting that the   Poisson pmf  overestimates  the tails of the true underlying Binomial distribution. In the continuous setting,  the comparison density estimator deviates mildly from one, suggesting that the  polynomial pdf may overestimate  the right tail of the truncated normal.   Finally, the upper panels show how, in virtue of the decomposition in \eqref{skewGest}, $\widehat{d}$ automatically  updates $g$ in the direction of $f$.


\section{Inference and graphics via LP-smoothed bootstrap}
\label{LPboot}
Despite the graph of $\widehat{d}_m(u;G,F)$ in \eqref{gajek} (or more broadly, of a suitable estimator $\widehat{d}$ of $d$) allows us to explore the nature of the deviation of $f$ from $g$, it does not provide any insight on the significance of such deviations. Conversely, smooth tests based on \red{the test statistics in \eqref{v} or \eqref{deviance}}, implicitly aim to test
{\fontsize{3.5mm}{3.5mm}\selectfont{
\begin{equation}
\label{test}
H_0: d(u;G,F) = 1\quad \text{for all $u \in[0,1]$}\qquad\text{vs}\qquad H_1 : d(u;G,F) \neq 1 \quad \text{for some $u \in[0,1]$},
\end{equation}}}
while in practice we test for
{\fontsize{3.5mm}{3.5mm}\selectfont{
\begin{equation}
\label{test2}
H_0: LP_1=\dots=LP_m= 0 \qquad\text{vs}\qquad H_1 : \text{at least one $LP_j\neq0$, $j=1,\dots,m$.}
\end{equation}}}
Notice that $H_0$ in \eqref{test} implies $H_0$ in \eqref{test2}, the opposite is not true in general. Whereas, $H_1$ in \eqref{test2} does imply $H_1$ in \eqref{test}, and thus smooth tests allow us to determine  if $f$ deviates significantly from $g$. However, they do not assess  where and  how significant departures of $f$ from $g$ occur. Therefore, to gain more insights on this aspect, in the next section we discuss how to complement the graph of $\widehat{d}_m$ with suitable confidence bands and graphically assess the validity of $H_0$ in \eqref{test}. 

\begin{algorithm}[!h]
\label{CIalgo}
\caption{Computing confidence bands and deviance tests via Monte Carlo.}
{\fontsize{3.3mm}{3.3mm}\selectfont{
\begin{tabbing}
{INPUTS:} sample observed $\bm{x}=(x_1,\dots,x_n)$, parametric start $g$, significance level $\alpha$, \\
\quad\qquad\qquad number of LP score functions $m$, number of Monte Carlo replicates $B$.\\
{Step 1:} Estimate $\widehat{LP}_{j}$, $j=1,\dots,m$ via \eqref{LPj} on $\bm{x}$.\\
{Step 2:} Compute \eqref{deviance} and call it $D_{\text{obs}}$.\\
{Step 3:} For $b=1,\dots, B:$\\
\qquad\qquad{A.} Sample $\bm{x}^{(b)}_G$ from $G$; \red{$\bm{x}^{(b)}_G$ must be of the same size as $\bm{x}$.} \\
\qquad\qquad{B.} On $\bm{x}^{(b)}_G$, compute  $\widehat{LP}^{(b)}_{j}$, $D^{(b)}$, $\widehat{d}^{(b)}_{m}$ via \eqref{LPj}, \eqref{deviance}, \eqref{gajek}, respectively.\\
{Step 4:} For each $u \in [0,1]$:\\
\qquad\qquad{A.} $\hat{\bar{d}}_{m}(u;G,F)=\frac{1}{B}\sum_{b=1}^B\widehat{d}^{(b)}_{m}(u;G,F)$\\
\qquad\qquad{B.} $SE_{\widehat{d}_m}(u|H_0)=\frac{1}{B}\sum_{b=1}^B\bigl(\widehat{d}^{(b)}_{m}(u;G,F)-\hat{\bar{d}}_{m}(u;G,F)\bigl)^2$\\
\qquad\qquad{C.} $\Delta(u)^{(b)}=\Bigl|\frac{\widehat{d}^{(b)}_{m}(u;G,F)-1}{SE_{\widehat{d}_m}(u|H_0)}\Bigl|$\\
{Step 5:} Estimate the quantile of order $1-\alpha$ of the distribution of $\max_u\Delta(u)$, i.e.,\\
\qquad\qquad\qquad $c_\alpha=\Bigl\{c: \frac{1}{B}\sum_{b=1}^BI\{\max_u\Delta(u)^{(b)}\geq c\} = \alpha \Bigl\} $.\\
{Step 6:} Combine Step 3B and Step 4 and compute \eqref{CIband}.\\
{Step 7:} Estimate the deviance test p-value via  $P(D\geq D_{\text{obs}}|H_0)=\frac{1}{B}\sum_{b=1}^BI\{D^{(b)}\geq D_{\text{obs}}\}$.
\end{tabbing}}}
\end{algorithm}
\subsection{Confidence bands and CD-plot}
\label{confidence}
When constructing confidence bands, we must take into account that their center and their width are determined by the bias and the variance of the comparison density estimator considered. In this section and those to follow, we focus on the estimator in \eqref{gajek}; however, \red{our considerations can be easily extended to any other estimator of $d$ which can be shown to be unbiased under $H_0$}.

Our $\widehat{d}_m$ estimator only accounts for the first $m+1$ LP score functions. Therefore, unless one were to assume that the true model is a special case of \red{the one specified in} \eqref{skewG} \citep[e.g.,][]{neyman37}, $\widehat{d}_m$ is a biased estimator of $d$ and confidence bands constructed around $\widehat{d}_m$ can potentially be shifted away from the true comparison density $d$. Although the bias cannot be easily quantified in a general setting, it is easy to show that, when $H_0$ in \eqref{test} (and consequently in \eqref{test2}) is true, $\widehat{d}_m$ is an unbiased estimator of \red{$d\equiv 1$} \citep[e.g.,][]{algeri20}. Hence, we can exploit this property to construct simultaneous confidence bands under the null hypothesis.
Furthermore, the variance of $\widehat{d}_m(u;G,F)$ is likely to vary substantially over the range $[0,1]$. For instance, when dealing with moderate sample sizes, it is natural to expect its standard error to be particularly large around the tails of the distribution.

To account for the issues associated with both bias and variance, we aim to construct confidence bands of the form
\begin{equation}
\label{CIband}
CI_{1-\alpha,H_0}(u)=\Biggl[1-c_\alpha SE_{\widehat{d}_m}(u|H_0),1+c_\alpha SE_{\widehat{d}_m}(u|H_0)\Biggl],
\end{equation}
where $SE_{\widehat{d}_m}(u|H_0)$ denotes the standard error of $\widehat{d}_m(u;G,F)$ under $H_0$ at $u$ and, letting $\alpha$ be the desired significance level, $c_{\alpha}$, is the value that satisfies
{\fontsize{3.5mm}{3.5mm}\selectfont{
\begin{equation}
\label{significance}
\begin{split}
1-\alpha&=P\Bigl(1-c_{\alpha}SE_{\widehat{d}_m}(u|H_0)\leq \widehat{d}(u;G,F)\leq 1+c_{\alpha}SE_{\widehat{d}_m}(u|H_0),\text{ for all $u\in[0,1]$}\Bigl|H_0\Bigl)\\
&=P\biggl(\max_{u} \biggl|\frac{\widehat{d}_m(u;G,F)-1}{SE_{\widehat{d}_m}(u|H_0)}\biggl|\leq c_{\alpha}\biggl|H_0\biggl).\\
\end{split}
\end{equation}}}
Despite tube-formulas can be used to approximate \eqref{CIband}  in an asymptotic, continuous regime \citep[e.g.,][]{algeri20}, such approach does not apply to the  discrete setting, when  the sample size is not  sufficiently large to rely on asymptotic approximations, or when model selection is performed (see Sections \ref{extensions} and \ref{COVID}). Therefore, to guarantee the generalizability of our procedure to all these situations, we rely on simulation methods to estimate $SE_{\widehat{d}_m}(u|H_0)$, $c_{\alpha}$, and construct the confidence bands in \eqref{CIband}.

\begin{figure*}[htb]
\begin{tabular*}{\textwidth}{@{\extracolsep{\fill}}@{}c@{}c@{}}
\hspace{-0.8cm}\includegraphics[width=79mm]{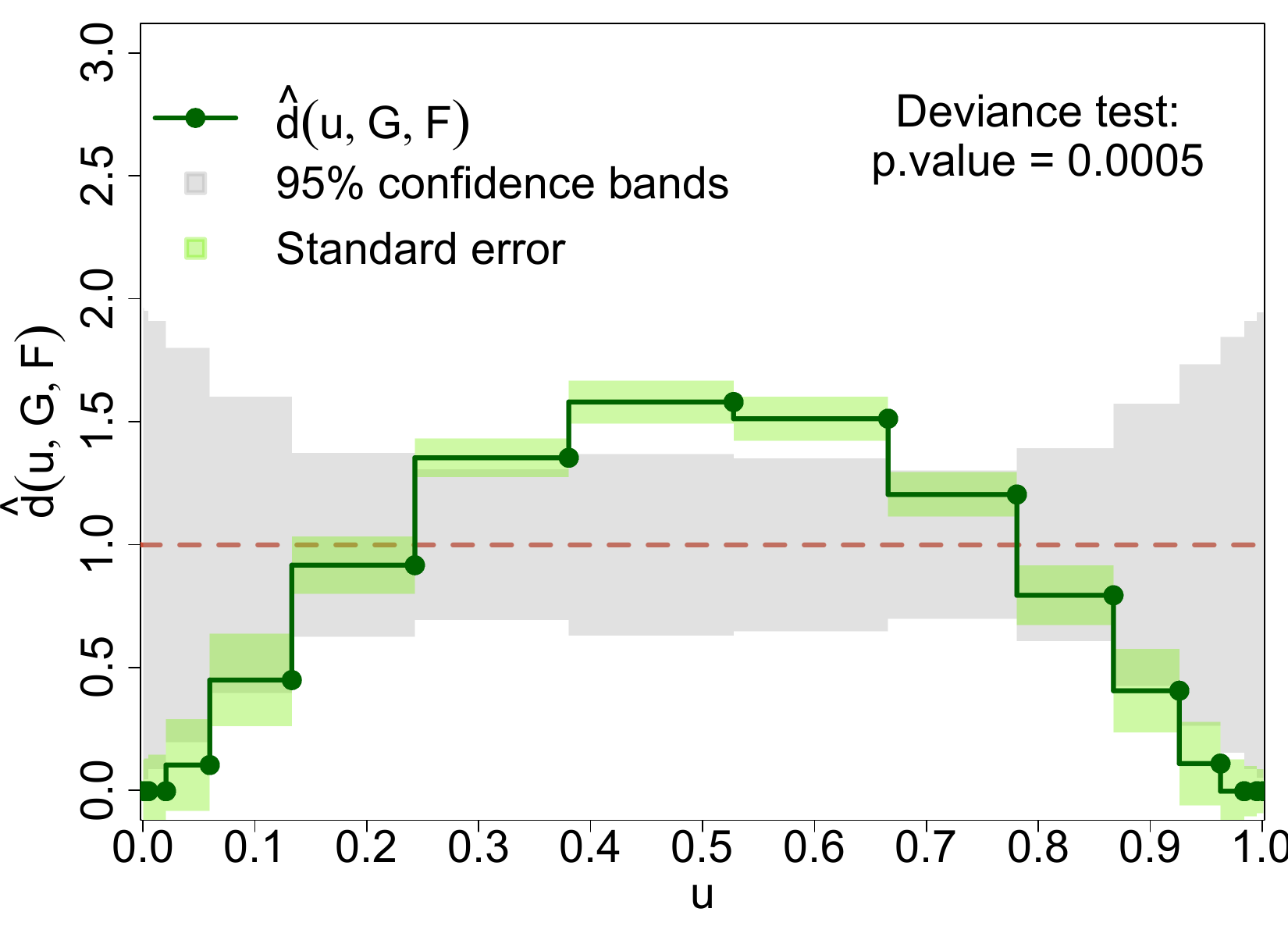} & \includegraphics[width=79mm]{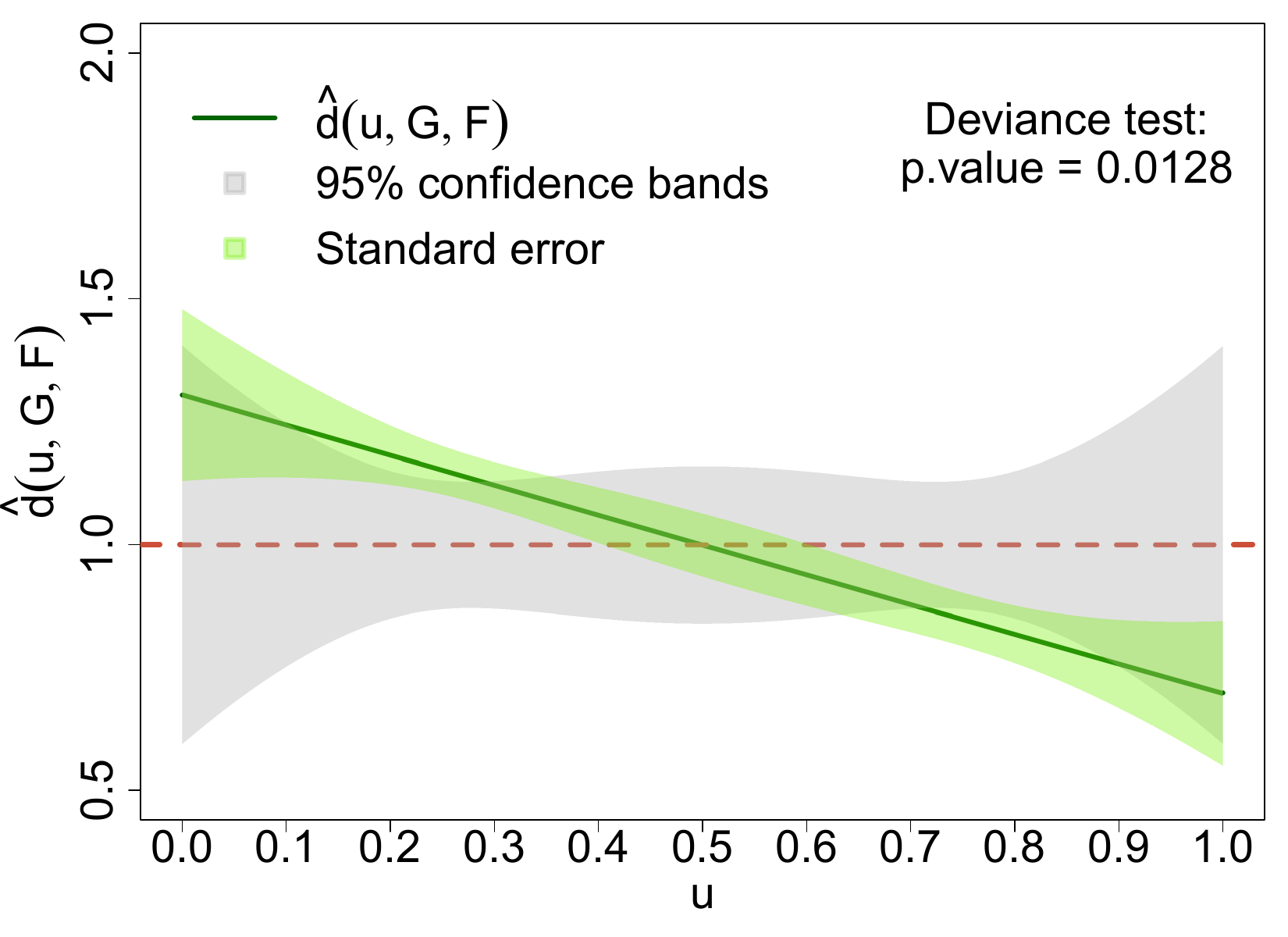}\\
\end{tabular*}
\caption{Examples of CD-plots when $X$ is discrete (left panel) and when $X$ is continuous (right panel). Similarly to Figure \ref{fig1}, in the left panel the hypothesized model is the pmf of a $\text{Poisson}(7.5)$ random variable truncated over the range $[0,15]$, whereas the data sample ($n=50$) is generated from a $\text{Binomial(15,0.5)}$ random variable. The hypothesized model in the right panel is the polynomial density in \eqref{poly}, whereas the data sample ($n=300$) is generated from a $\text{Normal(-15,15)}$ random variable truncated over the range $[0,30]$. The Gajek estimator of the comparison density in \eqref{gajek} is displayed using green lines with dots corresponding to the probability integral tranform of the mass points in the discrete case. In both cases, $m=2$, the gray bands correspond to the confidence bands obtained by simulating from the parametric start whereas the green bands are the standard errors obtained by simulating from Gajek's estimator in \eqref{gajek}. }
\label{fig2}
\end{figure*}
\subsection{Variance estimation via LP-smoothed bootstrap}
\label{LPscheme}
To compute the confidence bands in \eqref{CIband} via simulations, we estimate $SE_{\widehat{d}_m}(u|H_0)$ and $c_\alpha$ by drawing $B$ (e.g., if $\alpha=0.05$ set $B=10,000$ or $100,000$) Monte Carlo samples from $G$, namely
$\bm{x}^{(b)}_G$, $b=1,\dots,B$. Similarly, an approximate p-value to test \red{the hypotheses in} \eqref{test2} can be obtained by simulating the distribution of   \red{the deviance statistics in }\eqref{deviance}. The main steps of the simulation procedure are summarized in Algorithm 1.
The left panel of Figure \ref{fig2} shows the results obtained for our binomial example introduced in Section \ref{framework}. This plot is an example of \emph{Comparison Density plot} or \emph{CD-plot} which offers the advantage of visualizing where significant departures of the data distribution from the hypothesized parametric start occur. Specifically, if $\widehat{d}_m(u;G,F)$ is within the confidence bands (gray areas) over the unit interval, we conclude that there is no evidence that $f$ deviates significantly from $g$ anywhere over the range considered. Conversely, we expect significant departures to occur in regions where $\widehat{d}_m(u;G,F)$ lies outside the confidence bands. 
In our example, both the confidence bands as well as the p-values for the deviance test in \eqref{deviance} have been computed considering $B=10,000$ datasets simulated from the $\text{Poisson}(7.5)$ parametric start.

\red{Despite the bands in \eqref{CIband} only being affected by  the distribution of $\widehat{d}(u;G,F)$ under $H_0$, it is important to acquire a sense  sense of how small the variance of $\widehat{d}_m(u;G,F)$ is in the more general scenario where $G\not\equiv F$ in the more general scenario where $G\not\equiv F$. Specifically, despite the bias of $\widehat{d}_m(u;G,F)$ cannot be easily quantified in our setting, we can rely on simulation methods to obtain an estimate of its variance. Further considerations on the bias-variance trade-off in the context of estimation of linear functionals of $F$ are postponed to Section \ref{comparison}. }

As recognized by \citet{LPmode} (see also \citet{parzen2004}), \red{the decomposition in} \eqref{skewGest} naturally leads to a smoothed bootstrap scheme where samples from a smooth estimator of $f$ are obtained using an accept/reject algorithm where $g$ plays the role of the instrumental distribution. Conversely from the nonparametric bootstrap, in the smoothed bootstrap \citep{efron79} the sampling is performed from a smoothed version of $\tilde{F}$, and thus it avoids producing samples with several repeated values from the original data. \red{This is typically done using, for instance, kernel density estimators and which are known to be biased in the general setting. Despite that,  for suitable level of smoothing, the smoothed bootstrap may lead to   estimators with a lower mean square error (MSE) than those obtained with the nonparametric bootstrap \citep[e.g.][]{smooth1}.
Here, we focus on introducing an LP-smoothed bootstrap based on \red{the Gajek estimator in} \eqref{gajek} and we discuss its usefulness in estimating the variance of $\widehat{d}_m(u;G,F)$. A more detailed discussion on the advantages of using of the LP-smoothed bootstrap over the nonparametric bootstrap is postponed to Section \ref{comparison}. 
}

Denote with $\widehat{F}_m$ the estimate of $F$ associated with \red{the Gajek estimator in} \eqref{gajek} and let $x_G$ and $v$ be observations simulated from $G$ and $V\sim\text{Uniform}[0,1]$, respectively. We accept $x_G$ as an observation from $\widehat{F}_m$, i.e., we set $x_G=x_{\widehat{F}_m}$, if
\begin{equation}
\label{accept}
vM< \widehat{d}_m\bigl(G(x_G);G,F\bigl)\qquad\text{with $M=\max_{u\in [0,1]}\{\widehat{d}_m(u;G,F)\}$, }
\end{equation}
and we reject $x_G$ otherwise.
From a practical perspective, a smoothed bootstrap scheme based on \red{the accept/reject sampling scheme in} \eqref{accept} allows us to sample
from cells (or values) which have not been observed in the original data; therefore, it is particularly advantageous when dealing with discrete or categorized data and/or, more broadly, when the sample size is small (e.g., upper left panel of Figure \ref{fig1}). Moreover, expressing $\widehat{d}_m\bigl(G(x_G);G,F\bigl)$ using LP score functions naturally provides a smoothed bootstrap scheme that automatically generalizes to both the continuous and discrete settings.

The green bands in the left panel of Figure \ref{fig2} correspond to the estimated standard error of $\widehat{d}_m\bigl(G(x_G);G,F\bigl)$, namely $SE_{\widehat{d}_m}(u)$, for our binomial example and obtained by simulating $B=10,000$ datasets from $\widehat{F}_m$ via \red{the sampler in} \eqref{accept}.
In this setting, the parametric start is the pmf of a Poisson random variable and thus it is straightforward to first simulate observations from it, use them to compute \red{the confidence bands in} \eqref{CIband} and then accept/reject them as observations from $\widehat{F}_m$ to compute $SE_{\widehat{d}_m}(u)$. In practical applications, however, the hypothesized model $G$ does not always enjoy a simple formulation and thus \red{the accept/reject algorithm in} \eqref{accept} must be extended further to sample efficiently from both $G$ and $\widehat{F}_m$ as described in Section \ref{double}.

\subsubsection{A brief comparison with the nonparametric bootstrap}
\label{comparison}
In our setting, the level of smoothing is determined by $m$. As widely discussed by Young and co-authors \citep{smooth1, smooth2, smooth3, smooth4}, in the continuous case an adequate amount of smoothing may lead to estimators with a lower mean square error (MSE) than those obtained with the nonparametric bootstrap; however, these corrections are only up to the second order for large samples. In the discrete case, however, \citet{guerra} emphasize the advantages of the smoothing bootstrap in constructing confidence intervals but do not investigate the amount of smoothing required.

Despite an extensive comparison of the LP-smoothed bootstrap and the classical nonparametric bootstrap being beyond the scope of this paper, here we briefly discuss its advantages in the estimation of linear functionals of the type
\begin{equation}
\label{linear}
A(F)=< a,f>=\begin{cases}
\int a(x)f(x)\text{d} x &\quad\text{if $X$ is continuous,}\\
\sum_{r=1}^Ra(x_r)f(x_r) &\quad\text{if $X$ is discrete.}\\
\end{cases}
\end{equation}
Extensions to more general functionals can be derived as in \citet[][Sec. 3]{smooth1}.

A quantile representation of $A(F)$ is
\begin{equation}
\label{BD}
A(F)= < a,g\cdot d>= \int_0^1 b(u)d(u;G,F)\text{d} u,
\end{equation}
with $b(u)=a(G^{-1}(u))$.
Similarly, the estimates of $A(F)$ obtained by means of the classical nonparametric bootstrap and the LP-smoothed bootstrap, namely $A(\tilde{F})$ and $A(\hat{F}_m)$, respectively, can be specified as
\begin{equation}
\label{Aest}
A(\hat{F}_m)=\int_0^1 b(u)\widehat{d}_m(u;G,F)\text{d} u\quad \text{and} \quad A(\tilde{F})=\frac{1}{n}\sum_{i=1}^n a(X_i)=\int_0^1 b(u)\tilde{d}(u;G,F)\text{d} u,
\end{equation}
where $\tilde{d}$ is the empirical comparison density in \eqref{ecd}. Denote with $\dot{F}_m$ the cdf of \eqref{barton}, interestingly, when $\dot{f}_m$ is bona fide, $A(\tilde{F})=A(\widehat{F}_{\tilde{m}})=A(\dot{F}_{\tilde{m}})$. This aspect allows us to establish Theorem \ref{theo1} below; the proof is provided in Appendix 1.
\begin{theorem}
\label{theo1}
If $\dot{f}_m$ in \eqref{barton} is bona fide, $A(\widehat{F}_m)=A(\dot{F}_m)$ and the MSE of $A(\widehat{F}_m)$ can be lowered below that of $A(\tilde{F})$, for any $m\in \mathcal{M}=\{m: C(m)>0\}$ with
\begin{equation}
\label{delta}
C(m)=\sum_{j=m+1}^{\tilde{m}}\biggl[\frac{\sigma_{jj}}{n}-LP_j^2\biggl]B_j^2-2\sum_{m+1\leq j<k}^{\tilde{m}}\biggl[\frac{\sigma_{jk}}{n}-LP_jLP_k\biggl]B_jB_k,
\end{equation}
where $B_j=\int_0^1b(u)S_j(u;G)\text{d} u$, for all $j=1,2,\dots,{\tilde{m}}$ and $\tilde{m}=\min\{R-1,n-1\}$.
\end{theorem}
\begin{algorithm}[h]
\label{algo1}
\caption{Bidirectional acceptance sampling}
{\fontsize{3.3mm}{3.3mm}\selectfont{
\begin{tabbing}
INPUTS: sample $x_1,\dots,x_n$, parametric start $g$, instrumental probability function $h$.\\
Step 1: Obtain $\widehat{d}_m(G(x);G,F)$ and $\widehat{f}_m(x)$ in \eqref{gajek}.\\
Step 2: Set $d(H(x);H,G)=\frac{g(x)}{h(x)}$ and $d(H(x);H,\widehat{F}_m)=\frac{\widehat{f}_m(x)}{h(x)}$ and obtain $M^*$ in \eqref{Mpluminus}.\\
Step 3: Sample $x_H$ from $H$ and $v$ from $\text{Uniform}[0,1]$:\\
\quad\qquad a.{if } $x_H\in D^-\text{ and}\begin{cases}
\text{if} & vM^*\leq d(H(x_H);H,\widehat{F}_m)\Rightarrow \text{ set }x_H=x_{\widehat{F}_m}=x_G;\\
\text{else if} &d(H(x);H,\widehat{F}_m)<vM^*\leq d(H(x);H,G)\Rightarrow \text{ set } x_H=x_G;\\
\text{else} &\Rightarrow \text{ reject $x_H$};\\
\end{cases}$\\
\qquad \quad b.{ else} $x_H\in D^+\text{ and}\begin{cases}
\text{if}& vM^*\leq d(H(x);H,G)\Rightarrow \text{ set } x_H=x_G=x_{\widehat{F}_m};\\
\text{else if} &d(H(x);H,G)<vM^*\leq d(H(x);H,\widehat{F}_m)\Rightarrow \text{ set } x_H=x_{\widehat{F}_m};\\
\text{else} &\Rightarrow \text{ reject $x_H$}.\\
\end{cases}$
\end{tabbing}}}
\end{algorithm}
Because \red{the  criterion in} \eqref{delta} depends on unknown quantities, $C(m)$ can be estimated by replacing $LP_j$ and $\sigma_{jk}$ with consistent estimators. Furthermore, it has to be noted that the estimation of $C(m)$ may lead to numerical issues for large $\tilde{m}$, but it is feasible for $\tilde{m}<<\infty$. It follows that, in practice, Theorem \ref{theo1} is only useful for sufficiently large discrete-valued samples and small $R$. Conversely, for large continuous-valued samples, a suitable level of smoothing can be identified as in \citet{smooth3} or \cite{smooth2} and noticing that $\dot{d}_m(u;G,F)$ enjoys a kernel representation with bandwidth parameter proportional to $m^{-1}$ \citep{LPmode}.

\begin{remark} Notice that Theorem \ref{theo1} is only valid when $\dot{f}_m$ is bona fide. When that is not the case, $\widehat{f}_m\neq \dot{f}_m$ and the equivalent of Theorem \ref{theo1} cannot be easily derived. However, it can be shown that the MSE of $A(\dot{F}_m)$ approximates that of
$A(\widehat{F}_m)$ when $<a,g\cdot (\dot{d}_m^--K)>$ approaches zero, with $\dot{d}_m^-=\max\{0,-\dot{d}_m\}$, and thus \red{the  criterion in} \eqref{delta} can still be used as an approximate criterion to identify suitable values of $m$.
\end{remark}


\section{The bidirectional acceptance sampling}
\label{double}
The \emph{bidirectional acceptance sampling} extends  \red{the rejection sampling} in\eqref{accept} by sampling
simultaneously from both $\widehat{F}_m$ and $G$. The idea at the core of the algorithm is to consider an instrumental probability function, $h$ with cdf $H$, from which it is easy to sample.
 Samples from $H$ are then accepted/rejected as samples from both $G$ and $\widehat{F}_m$, from $G$ or $\widehat{F}_m$ only or from neither $G$ nor $\widehat{F}_m$. The main steps of this algorithm are described below and are summarized in Algorithm 2.

In principle, given pairs of observations $(x_H,v)$ drawn from $X_H\sim H$ and $V\sim \text{Uniform}[0,1]$, respectively, samples   $G$ and/or $\widehat{F}_m$ can be obtained as in \red{the sampling scheme in} \eqref{accept}, i.e.,
\begin{align}
\label{accept2a}
\text{if }& vM_G<d\bigl(H(x_H);H,G\bigl)\Rightarrow \text{set $x_H=x_G$, reject $x_H$ otherwise;}\\
\label{accept2b}
\text{if }& vM_{\widehat{F}_m}<d\bigl(H(x_H);H,\widehat{F}_m\bigl)\Rightarrow \text{set $x_H=x_{\widehat{F}_m}$, reject $x_H$ otherwise;}
\end{align}
with $d(H(x);H,G)=\frac{g(x)}{h(x)}$, $d(H(x);H,\widehat{F}_m)=\frac{\widehat{f}_m(x)}{h(x)}$, $M_G=\max_{u\in [0,1]}\{d(u;H,G)\}$ and $M_{\widehat{F}_m}=\max_{u\in [0,1]}\{d(u;H,{\widehat{F}_m})\}$.
Interestingly, \red{the  samplers in} \eqref{accept2a} and \eqref{accept2b} can be easily combined by noticing that
\begin{equation}
\label{dgf}
\widehat{d}_m(G(x);G,F)=\frac{d(H(x);H,\widehat{F}_m)}{d(H(x);H,G)},
\end{equation}
where $d(H(x);H,\widehat{F}_m)$ and $d(H(x);H,G)$ are known exactly and play the role of \emph{auxiliary} comparison densities. Specifically, let
\begin{equation}
\label{Dpluminus}
D^+=\bigl\{x:\widehat{d}_m(G(x);G,F)\geq1\bigl\}\quad\text{and}\quad D^-=\bigl\{x:\widehat{d}_m(G(x);G,F)<1\bigl\},
\end{equation}
\red{that is, $D^+$ is the region of the support of $X$ where $G$ underestimates $F$ or it is exactly equal to $F$, whereas $D^-$ is the region where $G$ overestimates $F$. Moreover,} define
\begin{equation}
\label{Mpluminus}
\begin{split}
 M^*&=\max\{M^-,M^+\} \qquad\text{with}\\
M^+&=\max_{x\in D^{+}}d(H(x);H,\widehat{F}_m)\quad\text{and}\quad M^-=\max_{x\in D^{-}}d(H(x);H,G),
\end{split}
\end{equation}
\red{that is, $M^+$ is the maximum value of the comparison density between $H$ and $\widehat{F}_m$ and $M^-$  is the maximum value of the comparison density between $H$ and $G$. $M^*$ is the largest of the two.}
From \red{the samplers in} \eqref{dgf}, \eqref{Dpluminus} it follows that $d(H(x);H,\widehat{F}_m)\geq d(H(x);H,G)$ for any $x\in D^+$.
The opposite is true for $x\in D^-$. Therefore,
if $x_H\in D^+$, we have that $vM^*\leq d(H(x_H);H,G)$ if $vM^*\leq d(H(x_H);H,\widehat{F}_m)$ and thus $x_H$ is accepted as a sample from both $G$
and $\widehat{F}_m$. Conversely, if $vM^*\leq d(H(x_H);H,\widehat{F}_m)$ but $vM^*> d(H(x_H);H,G)$, then $x_H\in D^+$ is accepted as a sample from $\widehat{F}_m$
but not from $G$. Finally, $x_H\in D^+$ is rejected whenever $vM^*> d(H(x_H);H,\widehat{F}_m)$. Similar reasoning applies for the case of $x_H\in D^-$.

\begin{theorem}
\label{theo2}
Let $\lambda_{\widehat{F}_m}=<\widehat{f}_m(x),I{\{x\in D^-\}}>$ and $\lambda_{G}=<g(x),I{\{x\in D^-\}}>$.
The bidirectional acceptance sampling enjoys the following properties:
\begin{enumerate}
\item it allows us to sample from $\widehat{F}_m$ and $G$,
\item its acceptance rate is $\frac{1}{M^*}$ for both $\widehat{F}_m$ and $G$,
\item for every $N$ observations drawn from $H$, the total number of evaluations of $d(\cdot;H,G)$ and $d(\cdot;H,\widehat{F}_m)$ is always less than or equal to $2N$ and it converges
almost surely to $2N-\Delta N$, with $\Delta=\frac{1}{M^{*}}[1+\lambda_{\widehat{F}_m}-\lambda_G]\in [0,\frac{1}{M^*}]$.
\end{enumerate}
\end{theorem}
Compared with \eqref{accept2a} and \eqref{accept2b}, Algorithm 2 reduces substantially the number of evaluations of $d(\cdot;H,G)$ and $d(\cdot;H,\widehat{F}_m)$ and can increase the acceptance rates for both $\widehat{F}_m$ and G. Theorem \ref{theo2} below summarizes these aspects and guarantees the validity of the bidirectional acceptance sampling. The proof of Theorem \ref{theo2} is given in the Supplementary Material.

\begin{figure*}[htb]
\begin{tabular*}{\textwidth}{@{\extracolsep{\fill}}@{}c@{}c@{}}
\hspace{-0.8cm}\includegraphics[width=78mm]{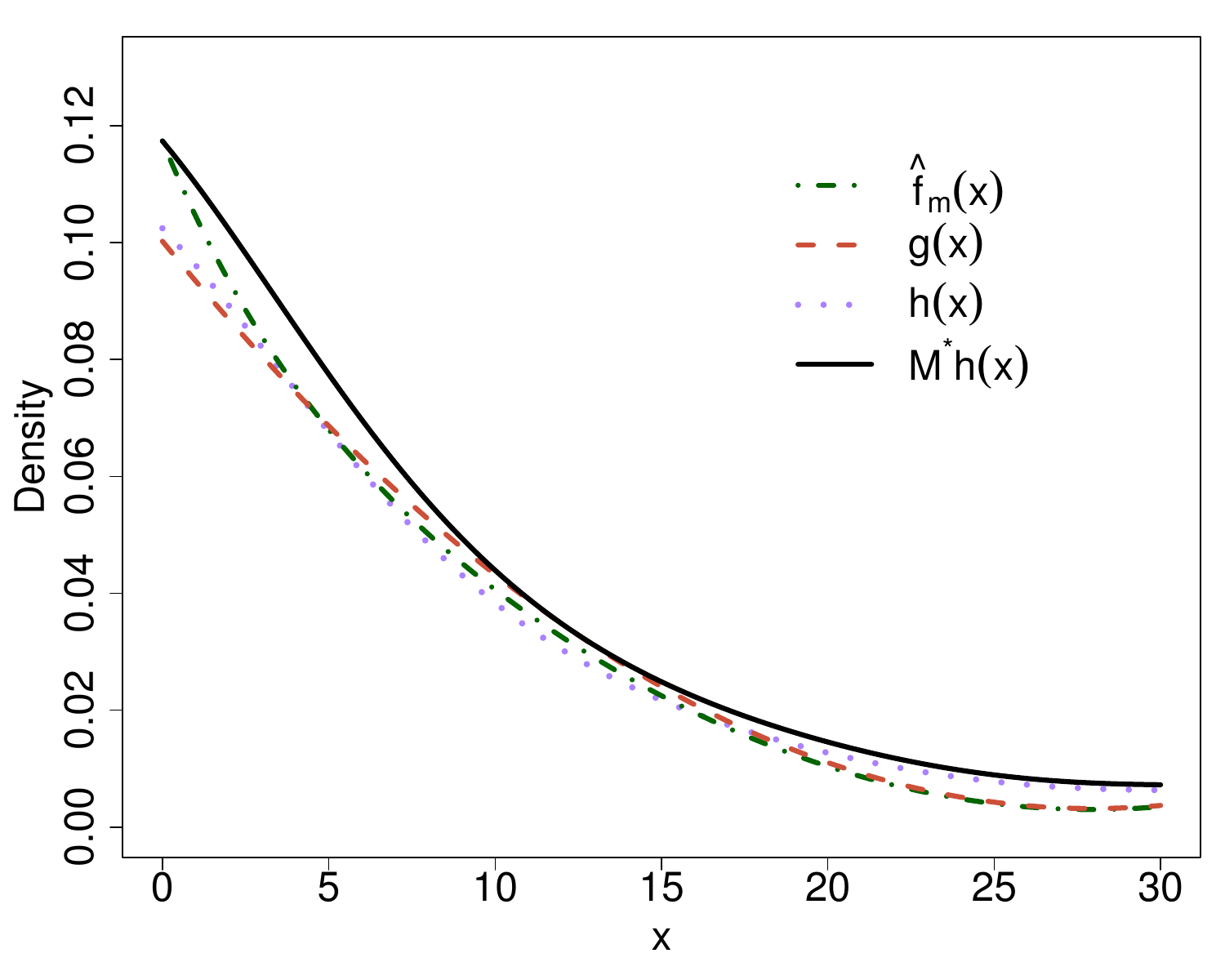} &\hspace{-0.1cm} \includegraphics[width=78mm]{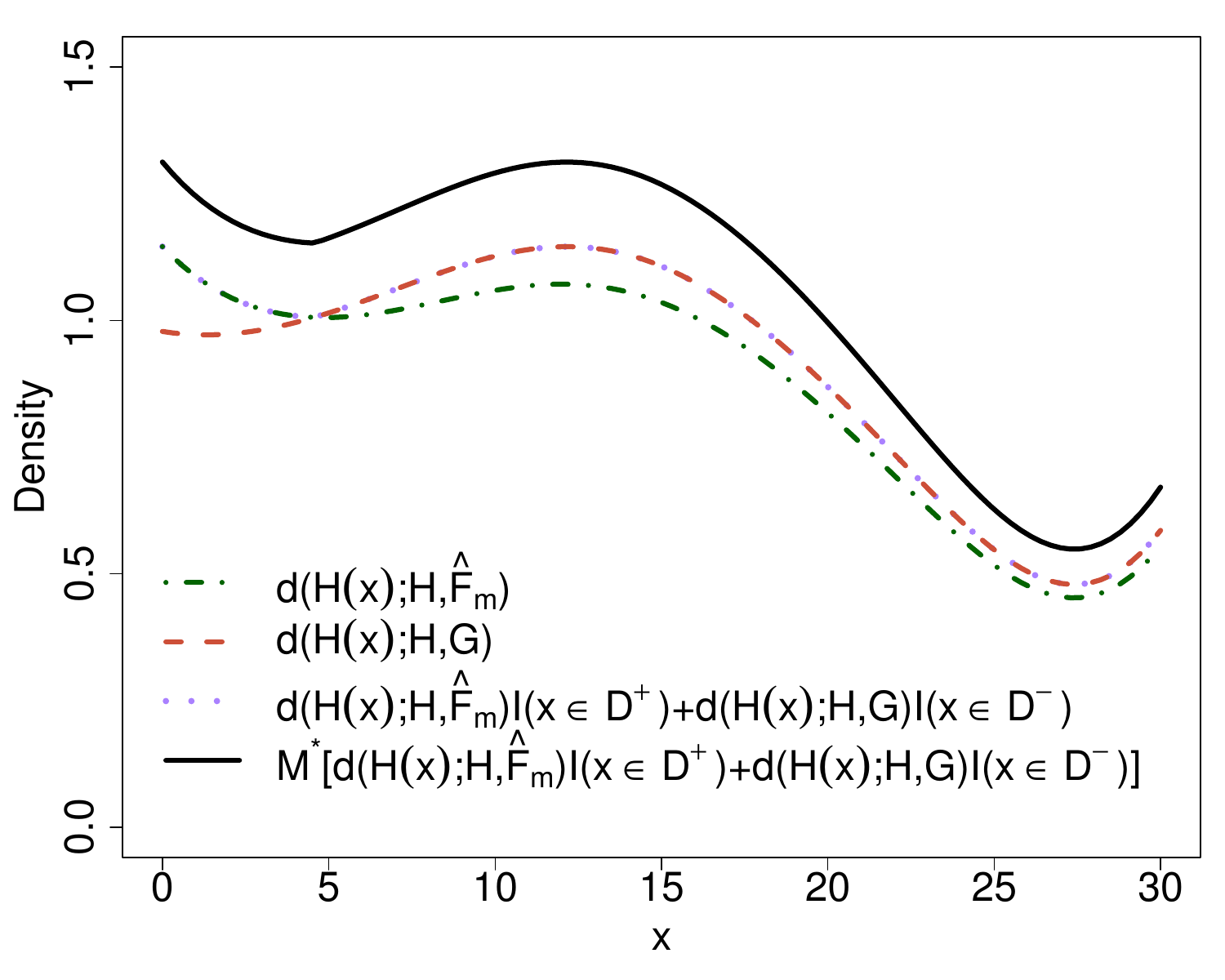}\\
\end{tabular*}
\caption{Choosing a suitable instrumental density $h$. The left panel shows how   $h$ in \eqref{mixture} (purple dotted line) compares with respect to the densities $g$ (red dashed line) and $\widehat{f}_m$ (green dotted--dashed line) from which we aim to simulate. The right panel shows the respective comparison densities with the purple dotted line corresponding to the function being optimized  in \eqref{choosingH}. }
\label{findingh}
\end{figure*}
The first property in Theorem \ref{theo2}   guarantees the validity of the algorithm, i.e., it ensures that the resulting samples are effectively drawn from the desired distributions $\widehat{F}_m$ and $G$. The second property ensures that the acceptance rate is the same for both $\widehat{F}_m$ and $G$ and, as discussed in more detail in Section \ref{choices}, it can be lowered below that of both \red{samplers in} \eqref{accept2a} and \eqref{accept2b} for suitable choices of $h$. Finally, the third property guarantees that the number of evaluations of $d(\cdot;H,G)$ and
$d(\cdot;H,\widehat{F}_m)$ is reduced by at least a factor of $N$ compared with \red{the accept/reject algorithms in} \eqref{accept2a} and \eqref{accept2b}. This is a substantial advantage in terms of efficiency of the algorithm, mostly when dealing with complex functional forms for $g$.

\begin{remark}
Notice that, in principle, the acceptance sampling algorithms in \eqref{accept} and Algorithm 2 can be implemented considering any smooth estimator of the form of \eqref{skewGest}. However, when considering estimators that are not bona fide, such as \eqref{barton}, it is easy to show that the resulting samples are
not drawn from $\dot{F}_m$ but rather by its bona fide counterpart in the formulation of Devroye--Gy\"orfi (see Table \ref{families}). Consequently,  its acceptance rate differs from $1/M^*$ by a multiplicative factor of $K_{\dot{f}_m^+}=<\dot{f}_m,I{\{\dot{f}_m> 0\}}>$. Similarly, when $g$ is only known up to a normalizing constant, its acceptance rate differs from $1/M^*$ by a multiplicative factor of $K_G=<g,1>$.
\end{remark}

\subsection{Choice of $h$}
\label{choices}
Despite the choice of the instrumental probability function $h$ in Algorithm 2 being arbitrary, it is should be simple enough so that it is (i) easy to sample from and (ii) sufficiently flexible to generalize to different settings. Possible choices of $h$ that satisfy these two criteria include  the density of a mixture of normal random variables when $X$ is continuous with unbounded support, 
 the density of a mixture of truncated normal random variables when $X$ is continuous with bounded support and the pmf of a mixture of Poisson and/or negative binomial random variables when $X$ is discrete.

Notice that $h$ plays a fundamental role because it affects both the acceptance rate and the computational efficiency of the sampling scheme. Therefore, once a parametric form for $h$ is selected, one can calibrate its parameters to minimize $M^*$, i.e.,
\begin{equation}
\label{choosingH}
h=\arg_{h}\min\{M^*\}=\arg_{h}\min\max_x\frac{g(x)}{h(x)}\biggl[\widehat{d}_m(G(x);G,F)I_{\{x\in D^+\}}+I_{\{x\in D^-\}}\biggl].
\end{equation}
Finally, the acceptance rate of $\widehat{F}_m$ and $G$ is reduced below that of \red{the  samplers in} \eqref{accept2a}-\eqref{accept2b} whenever $M^*\leq \min\{M_G, M_{\widehat{F}_m}\}$.

The right panel of Figure \ref{fig2} shows the CD-plot obtained for our truncated normal example. Confidence bands, deviance p-value and standard errors have been computed by simultaneously sampling $B=10,000$ datasets from $G$ and $\widehat{F}_m$ by means of Algorithm~2. In this case, the instrumental density $h$ is chosen to be the pdf of a mixture of truncated normals over $[0,30]$ with three components. The mixture weights, means and variances have been chosen to be the solutions of \eqref{choosingH} and lead to the density
\begin{align}
\label{mixture}
h(x)=&0.012\phi_{[0,30]}(34.919,5.694)+ 0.466 \phi_{[0,30]}( 6.251,11.953)\\
& + 0.522 \phi_{[0,30]}(-5.331,8.008),
\end{align}
where $\phi_{[0,30]}(\mu,\sigma)$ denotes the pdf of a normal distribution with mean $\mu$, standard deviation $\sigma$, and truncated over the range $[0,30]$. The resulting value for $M^*$ is $1.146$ which guarantees an acceptance rate of $87.284\%$ for both $G$ and $\widehat{F}_m$ and it is such that $M^*=1.146\approx M_{\widehat{F}_m}\approx M_G$. Therefore, it leads to approximately the same acceptance rate as \red{the accept/reject algorithms  in} \eqref{accept2a} and \eqref{accept2b}, 
\begin{algorithm}[!h]
\label{algo2}
\caption{LP-smoothed inference}
{\fontsize{3.3mm}{3.3mm}\selectfont{
\begin{tabbing}
{INPUTS:} sample observed $\bm{x}=(x_1,\dots,x_n)$, parametric start $g$, significance level $\alpha$, \\
\quad\qquad\qquad number of LP score functions $m$, number of Monte Carlo replicates $B$, \\
\quad\qquad\qquad instrumental probability function $h$ (optional).\\
{Step 1:} Estimate $\bm{\beta}$ via \eqref{MLE} on $\bm{x}$.\\
{Step 2:} Estimate $\widehat{LP}_{j,\widehat{\bm{\beta}}}$, $j=1,\dots,m$, via \eqref{LPjbeta} on $\bm{x}$ .\\
{Step 3:} Select the $\widehat{LP}_{j,\widehat{\bm{\beta}}}$ coefficients via \eqref{BIC} and set all the others to zero.\\
{Step 4:} Compute \eqref{deviance} and call it $D_{\text{obs}}$.\\
{Step 5:} Obtain $\widehat{d}_m$ and $\widehat{f}_m$ via \eqref{gajek_beta}.\\
{Step 6:} For $b=1,\dots, B:$\\
\qquad\qquad Obtain samples $\bm{x}^{(b)}_G$ from $G_{\widehat{\beta}}$ and $\bm{x}^{(b)}_{\widehat{F}_m}$ from $\widehat{F}_m$ via \eqref{accept} or Algorithm 2; \\
\qquad\qquad\red{$\bm{x}^{(b)}_G$ and $\bm{x}^{(b)}_{\widehat{F}_m}$ must be of the same size as $\bm{x}$.} \\
\qquad\qquad {A.} On $\bm{x}^{(b)}_G$:\\
\qquad\qquad\qquad{i.} Estimate $\bm{\beta}$ via \eqref{MLE} and call it $\widehat{\bm{\beta}}^{(b)}_{G}$.\\
\qquad\qquad\qquad {ii.} Estimate $\widehat{LP}_{j,\widehat{\bm{\beta}}_G}$, $j=1,\dots,m$, via \eqref{LPjbeta}.\\
\qquad\qquad\qquad{iii.} Set to zero nonsignificant $\widehat{LP}^{(b)}_{j,\widehat{\bm{\beta}}_G}$ coefficients via \eqref{BIC}.\\
\qquad\qquad\qquad {iv.} Compute $D^{(b)}$ via \eqref{deviance}.\\
\qquad\qquad\qquad {v.} Estimate $d$ via \eqref{gajek_beta} and call it $\widehat{d}^{(b)}_{m,{G}}(u)$.\\
\qquad\qquad {B.} On $\bm{x}^{(b)}_{\widehat{F}_m}$:\\
\qquad\qquad\qquad{i.} Estimate $\bm{\beta}$ via \eqref{MLE} and call it $\widehat{\bm{\beta}}^{(b)}_{\widehat{F}_m}$.\\
\qquad\qquad\qquad {ii.} Estimate $\widehat{LP}_{j,\widehat{\bm{\beta}}_{\widehat{F}_m}}$, $j=1,\dots,m$, via \eqref{LPjbeta}.\\
\qquad\qquad\qquad{iii.} Set to zero nonsignificant $\widehat{LP}^{(b)}_{j,\widehat{\bm{\beta}}_{\widehat{F}_m}}$ coefficients via \eqref{BIC}.\\
\qquad\qquad\qquad {iv.} Estimate $d$ via \eqref{gajek_beta} and call it $\widehat{d}^{(b)}_{m,{\widehat{F}_m}}$(u).\\
{Step 7:} For each $u \in [0,1]$:\\
\qquad\qquad{A.} $\hat{\bar{d}}_{m,G}(u;G,F)=\frac{1}{B}\sum_{b=1}^B\widehat{d}^{(b)}_{m,G}(u)$\\
\qquad\qquad{B.} $SE_{\widehat{d}_m}(u|H_0)=\frac{1}{B}\sum_{b=1}^B\bigl(\widehat{d}^{(b)}_{m,G}(u)-\hat{\bar{d}}_{m,G}(u)\bigl)^2$\\
\qquad\qquad{C.} $\Delta(u)^{(b)}=\Bigl|\frac{\widehat{d}^{(b)}_{m,G}(u)-1}{SE_{\widehat{d}_m}(u|H_0)}\Bigl|$\\
\qquad\qquad{D.} $\hat{\bar{d}}_{m,\widehat{F}_m}(u;G,F)=\frac{1}{B}\sum_{b=1}^B\widehat{d}^{(b)}_{m,\widehat{F}_m}(u)$\\
\qquad\qquad{E.} $SE_{\widehat{d}_m}(u)=\frac{1}{B}\sum_{b=1}^B\bigl(\widehat{d}^{(b)}_{m,\widehat{F}_m}(u)-\hat{\bar{d}}_{m,\widehat{F}_m}(u)\bigl)^2$\\
{Step 8:} Estimate the quantile of order $1-\alpha$ of the distribution of $\max_u\Delta(u)$, i.e.,\\
\qquad\qquad\qquad $c_\alpha=\Bigl\{c: \frac{1}{B}\sum_{b=1}^BI\{\max_u\Delta(u)^{(b)}\geq c\} = \alpha \Bigl\} $.\\
{Step 9:} Combine Step 7B and Step 8 and compute \eqref{CIband}.\\
{Step 10:} Estimate the deviance test p-value via  $P(D\geq D_{\text{obs}}|H_0)=\frac{1}{B}\sum_{b=1}^BI\{D^{(b)}\geq D_{\text{obs}}\}$.
\end{tabbing}}}
\end{algorithm}
while reducing the number of evaluations of the auxiliary densities. Figure \ref{findingh} compares $h$, $g$, $\widehat{f}_m$ and the respective comparison densities.

Finally, both the confidence bands as well as the p-values for the deviance test in \eqref{deviance} suggest that the polynomial parametric start in \eqref{poly} overestimates the tail on the distribution.

\section{Important extensions and further considerations}
\label{extensions}
For the sake of introducing the main elements of our framework, so far we have only considered situations where the parametric start $g$ is fully specified and $m$ is chosen arbitrarily. Here, we discuss how the constructs of Sections \ref{framework}--\ref{double} can be adequately extended to situations where $g$ depends on unknown parameters and $m$ is selected from a pool of possible values. Algorithm 3 summarizes the steps necessary to compute standard errors, confidence bands and deviance tests under this regime.

\subsection{Testing a composite hypothesis}
\label{composite}
In practical applications, the parametric start $g$ may depend on a set of free parameters, namely, $\bm{\beta}$ which need to be adequately estimated. This situation falls under the framework of \emph{smooth tests when the hypothesis is composite} \citep{barton56} and which can be tackled as follows.

Denote with $d(u;G_{\bm{\beta}},F)$ the comparison density between $F$ and $G$ when the latter depends on the unknown parameter $\bm{\beta}$. To obtain a suitable estimate of $d$ in this setting, we proceed by first estimating $\bm{\beta}$ on the observed sample $\bm{x}=(x_1,\dots,x_n)$ via Maximum Likelihood Estimation (MLE) and assuming $g$ to be the true model for the data, i.e.,
\begin{equation}
\label{MLE}
\widehat{\bm{\beta}}=\arg_{\bm{\beta}}\max \sum_{i=1}^n \log\{g(x_i,\bm{\beta})\}.
\end{equation}
An estimate of $d(u;G_{\bm{\beta}},F)$, can be obtained as described in Section \ref{framework} by simply setting $g= g_{\widehat{\bm{\beta}}}$ and $G= G_{\widehat{\bm{\beta}}}$. Specifically, letting $U=G_{\widehat{\bm{\beta}}}(X)$ the $\widehat{LP}_j$ estimates in \eqref{LPj} are replaced by their ``composite'' counterparts
\begin{equation}
\label{LPjbeta}
\widehat{LP}_{j,\widehat{\bm{\beta}}}=\frac{1}{n}\sum_{i=1}^nT_j(x_i;G_{\widehat{\bm{\beta}}})=\int_0^1S_j(u;G_{\widehat{\bm{\beta}}})\tilde{d}(u;G_{\widehat{\bm{\beta}}},F)\text{d} u,
\end{equation}
and thus, Gajek's estimator in \eqref{gajek} can be rewritten as

\begin{equation}
\label{gajek_beta}
\begin{split}
\widehat{f}_{m}(x)&=g_{\widehat{\bm{\beta}}}(x)\widehat{d}_{m}(G_{\widehat{\bm{\beta}}}(x);G_{\widehat{\bm{\beta}}},F)\quad\text{with}\\
\quad \widehat{d}_{m}(G_{\widehat{\bm{\beta}}}(x);G_{\widehat{\bm{\beta}}},F)&= \max{\biggl\{0,1+\sum_{j=1}^{m}\widehat{LP}_{j,\widehat{\bm{\beta}}}T_j(x;G_{\widehat{\bm{\beta}}})-K \biggl\}}.
\end{split}
\end{equation}
Similarly to Section \ref{LPboot}, the hypotheses being tested can be specified by replacing $d(u;G,F)$ with $d(u;G_{\bm{\beta}},F)$ in \red{the test in} \eqref{test} and $LP_j$ with $LP_{j,\bm{\beta}}$ in  \red{the test in} \eqref{test2}. Finally, adequate confidence bands, standard errors and deviance tests can be computed via  \red{the accept/reject sampling scheme in} \eqref{accept} and Algorithm 2, using  \red{the parametric Gajek estimator in} \eqref{gajek_beta} in place of \eqref{gajek} and estimating $\bm{\beta}$  \red{via MLE} as in \eqref{MLE} at each replicate as described in Steps 1, 6.A.i and 6.B.i of Algorithm 3.

\subsection{Data-driven smoothed inference}
\label{choose_m}
In principle, one could select the value of $m$ that appears to provide the best fit to the data. However, as extensively discussed in \citet{ledwina,kallenberg} (see also \citet[Ch.10]{rayner2009}) a poor choice of $m$ may lead to substantial loss of power. To overcome this problem, they introduce \emph{data-driven smooth tests} where the size of the orthonormal basis to be considered is selected using Schwartz's BIC criterion. A similar approach has been proposed by \citet{LPmode,LPcopula} in the context of LP modeling and, given an initial set of $m_\text{Max}$ coefficients, it consists in arranging them in decreasing magnitude i.e.,
$\widehat{LP}^2_{(1)}\geq \widehat{LP}^2_{(2)}\geq\dots\geq \widehat{LP}^2_{(m_\text{Max})}$, and selecting the largest $m$ for which $BIC(m)$ in \eqref{BIC} is maximum
\begin{equation}
\label{BIC}
BIC(m)=\sum_{(j)=1}^m\widehat{LP}^2_{(j)} - \frac{m \log n }{n}.
\end{equation}
\red{The remaining coefficients are set to zero. From a practical standpoint, by applying the methods discussed here to various settings, we have observed that the BIC criterion in \eqref{BIC} tends to lead to estimators which are either as smooth as the true density or, in some cases, overly smooth. This is particularly noticeable when compared with other selection criteria \citep[e.g.,][]{iGOF}. In general, however, one can also assess the reliability of the value of $m$ selected by the BIC criterion in \eqref{BIC} by visually inspecting the graph of the resulting estimator of $d$. Whereas, the criterion in Theorem \ref{theo1}, allows us to assess if the level of smoothing (in our case determined by the number $m$ of coefficients selected) is expected to lead to better estimators than those obtained using the classical non-parametric bootstrap. }

Unfortunately, because criteria such as \eqref{BIC} are data driven, they strongly affect the distribution of the $\widehat{LP}_j$ estimators; hence,  they introduce an additional source of variability one must account for in the variance estimation process while the inference must be adequately adjusted post-selection. In our setting, this can be easily done by repeating the selection process at each Monte Carlo and/or bootstrap replicate as summarized in Steps 3, 6.A.iii and 6.B.iii of Algorithm 3.

\begin{figure}[htb]
\begin{center}
\includegraphics[width=0.7\textwidth]{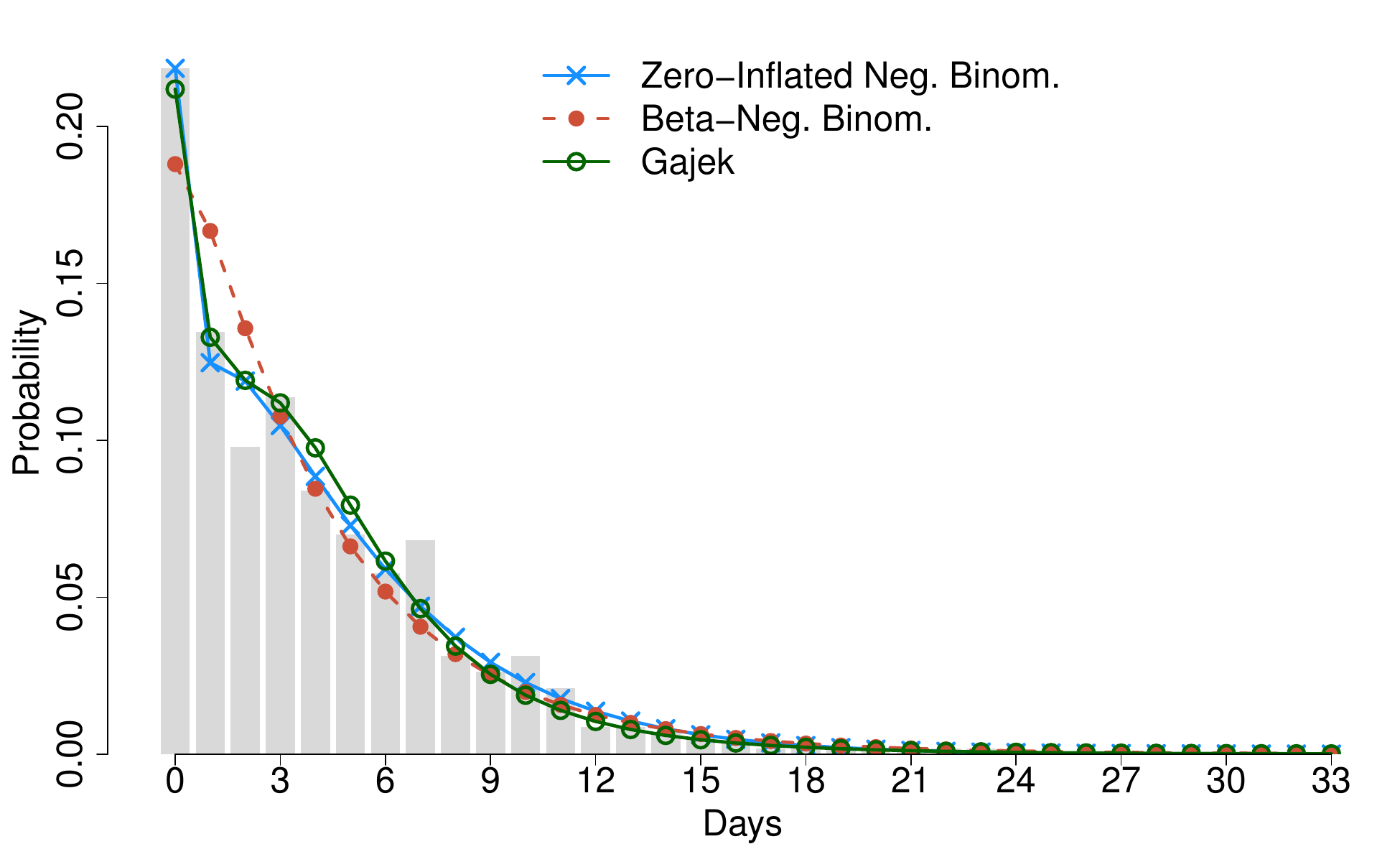}
\end{center}
\caption{Histogram of the data and fitted models. The red dots correspond to the pmf of a beta-negative binomial while the blue crosses are the pmf of a zero-inflated negative binomial. In both cases, the unknown parameters are estimated via MLE. The green circles correspond to Gajek's estimator in \eqref{gajek_beta}, with parametric start set to be the estimated beta-negative binomial model. }
\label{covidHist}
\end{figure}
\begin{figure*}[htb]
\begin{tabular*}{\textwidth}{@{\extracolsep{\fill}}@{}c@{}c@{}}
\hspace{-0.8cm}\includegraphics[width=79mm]{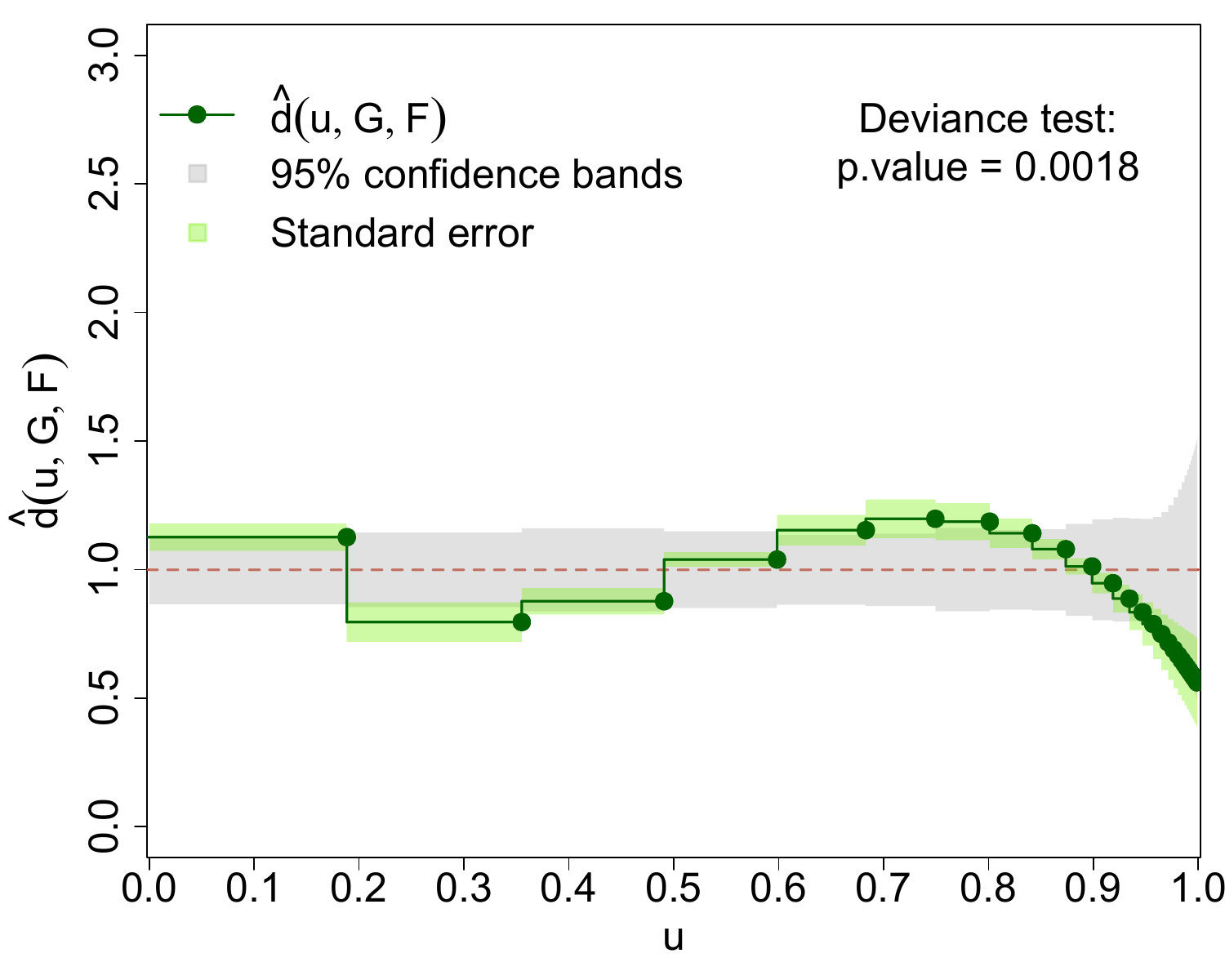} &  \includegraphics[width=79mm]{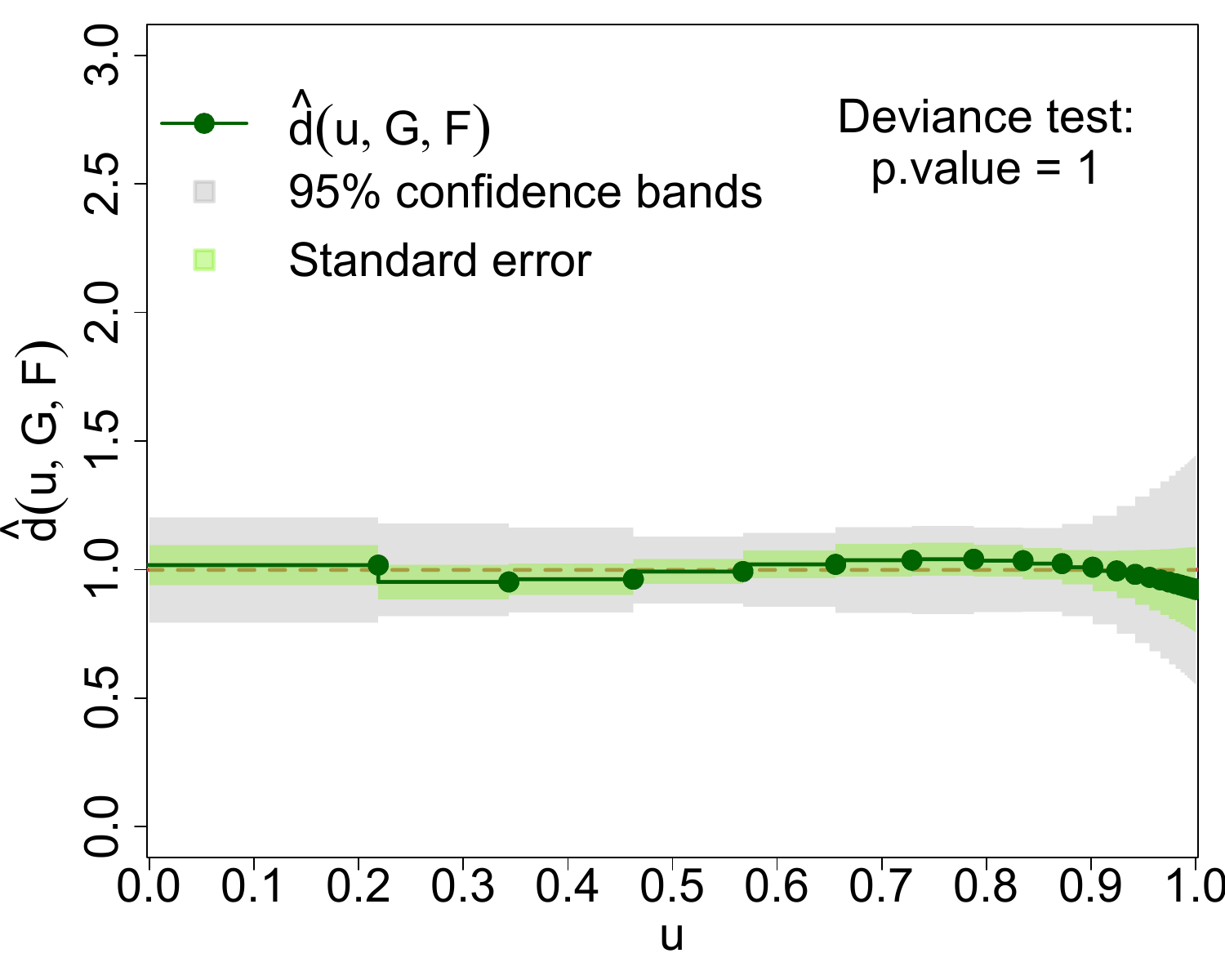}\\
\end{tabular*}
\caption{Deviance tests and CD-plots for beta-negative binomial \red{(BNB)} model (left panel) and zero-inflated negative binomial (ZINB) model (right panel). The comparison density estimated via \eqref{gajek_beta} is plotted using green dotted lines. The $LP_j$ coefficients have been selected via \red{the BIC criterion in} \eqref{BIC} from a pool of $m_{\text{Max}}=10$. The gray bands correspond to the $95\%$ confidence bands while the green bands refer to the standard errors. Confidence bands, standard errors and deviance tests have been computed via Algorithm 3.}
\label{CDplotsCovid}
\end{figure*}
\section{COVID-19 time of hospitalization from symptoms onset }
\label{COVID}
In December 2019, several cases of infections by a novel coronavirus named SARS-CoV-2 were reported in Wuhan, China. SARS-CoV-2 is responsible for a respiratory disease named COVID-19 whose symptoms may resemble those of the seasonal flu. Unfortunately, in some cases, COVID-19 can rapidly evolve into severe pneumonia, posing a severe risk to the survival of the affected patients. Today, COVID-19 has spread across the entire globe  causing \red{millions of fatalities}. Understanding the evolving epidemiology of COVID-19 is crucial in planning  health care resources \citep{garg}.

The analyses presented here   exploit the methods proposed in the previous sections to study the distribution of the time (in days) from symptoms onset to hospitalization of  patients. 
\red{Specifically, here we focus on a subset of the \texttt{COVID19\_open\_line\_list.csv} dataset provided  in \citet{xu2020} and  contains information of date of symptoms onset, date of hospitalization, age, gender and travel history of COVID-19 patients.  The original dataset was  downloaded by the authors in March 2020 from Kaggle\footnote{\texttt{https://www.kaggle.com/sudalairajkumar/novel-corona-virus-2019-dataset}}.  However, given the urge to acquire as much knowledge as possible on the ongoing COVID-19 pandemic, data related to the novel 2019 coronavirus are constantly updated on many data-sharing platforms.
Therefore, for the sake of the reproducibility of our results, in the Supplementary Material we provide the specific subset of the \texttt{COVID19\_open\_line\_list.csv} dataset, as of March 2020,  used to perform the analyses discussed in this section.  More recent versions of these data are available on
 the \texttt{github} data repository \texttt{nCoV2019}\footnote{\texttt{https://github.com/beoutbreakprepared/nCoV2019}}.}.  

We consider $n=572$ patients for which both the dates of hospitalization and  symptoms onset were recorded, and whose symptoms first appeared before or on the day of hospitalization. The histogram of the data is shown in Figure \ref{covidHist}.
Looking at the histogram of the data, it is easy to see that, for most of the patients, the date of onset symptoms coincides with the date of hospitalization.
From a statistical perspective, it is particularly interesting to understand how the excess of zeros should be modeled, as this may provide additional insights on the underlying cause.
For instance, this phenomenon may indicate that for many of the hospitalized patients, the symptoms increased severely within the next 24 hours since they first appeared. It is also possible, however, that two different processes may have simultaneously contributed to the excess zeros. For example, if the information regarding the date of symptoms onset was not available, the latter may have been recorded to be the same as the date of hospitalization. 
Here, we attempt to model the excess of zeros considering   a beta-negative binomial (BNB) distribution and a zero-inflated negative binomial (ZINB). The BNB and ZINB models will be considered in turn as a parametric start for our procedure.

First, we focus on the BNB assuming that such distribution can be particularly helpful to model a rapid decay of the number of days from symptoms onset and hospitalization. Unfortunately, when estimating the unknown parameters via MLE, the fitted distribution (red dots in Figure \ref{covidHist}) appears to underestimate the excess of zeros. Therefore, we construct the respective deviance test and CD-plot (see Algorithm 3) to assess where significant deviations from the true underlying model occur. The results obtained are collected in the left panel of Figure \ref{CDplotsCovid}. The deviance test rejects the BNB hypothesized and the CD-plot clearly shows that the most substantial departures occur in \red{proximity
of the first quartile of the hypothesized distribution $G$ and its right tail.}
Finally, an updated version of the BNB model is constructed using the Gajek estimator in \eqref{gajek_beta} and selecting the most ``significant'' $LP_j$ coefficients from a pool of $m_{\text{Max}}=10$ via \eqref{BIC}. The resulting estimate is
\begin{equation}
\label{bnb}
\widehat{f}_{(1)}(x)=g(x,33,11.098,1.218)[-0.159T_3(x,G_{11.098,1.218})],
\end{equation}
where $g(x,R,\xi,\nu)$ denotes the pmf of a BNB random variable with support $\{0,\dots,R\}$ and shape parameters $\xi$ and $\nu$, $G_{\xi,\nu}$ is the respective cdf. In this case, $R=33$, whereas the MLEs of $\xi$ and $\nu$ are 11.098 and 1.218, respectively. In $\widehat{f}_{(1)}(x)$, the subscript is used to denote that the model selection rule in \eqref{BIC} sets to zero all the $\widehat{LP}_j$ estimates with $j\neq 3$ and thus the only basis function considered is $T_3(x,G_{11.098,1.218})$ with $\widehat{LP}_3=-0.159$. The pmf estimate in \eqref{bnb} is plotted as green circles in Figure \ref{covidHist}. Interestingly, the estimated smoothed model suggests a better fit to the data and adequately accounts for the excess of zeros characterizing the underlying distribution.

Previous studies proposed in the literature on the analysis of hospitalization data have suggested that ZINB models are often preferred because they account for both zero-excesses and over dispersion \citep[e.g.,][]{weaver}. Here, we assess if a ZINB can be used to adequately model our data. Also in this case, the unknown parameters characterizing the ZINB model have been estimated via MLE while ``significant'' $LP_j$ coefficients have been selected as in \eqref{BIC} and setting $m_{\text{Max}}=10$. The estimated ZINB (blue crosses in Figure \ref{covidHist}) is very close to the estimator in \eqref{bnb}. Furthermore, both deviance test and CD-plot (see right panel Figure \ref{CDplotsCovid}) fail to reject the ZINB suggesting that the latter is a reliable model for our data.

\section{Discussion}
\label{discussion}
In this article, we combine smooth tests, smoothed bootstrap and LP modeling aiming to establish a unified framework for goodness-of-fit that naturally integrates modeling, estimation, inference and graphics.

As highlighted in Section \ref{framework}, LP modeling plays a crucial role to ensure the generalizability of the methods proposed. Specifically,  through the use of LP score functions  we   guarantee universality with respect to the continuous or discrete nature of the data under study. Furthermore, the comparison density is the key ingredient to simultaneously perform confirmatory and exploratory goodness-of-fit using a CD-plot. As shown in Section \ref{COVID}, the latter provides a detailed yet concise visualization of the discrepancies between the hypothesized and the true model in terms of quantiles of the distribution. 

The LP-based smoothed bootstrap scheme described in Section \ref{LPscheme} is the first step in the variance estimation process and easily applies to arbitrarily large samples. In this manuscript, we only briefly study its performance with respect to the nonparametric bootstrap. The reason being that, for our purposes, the smoothed bootstrap is ultimately used to adjust the estimate of the standard error of $\widehat{d}_m$ after a model selection procedure is implemented (see Section \ref{extensions}). Results such as Theorem \ref{theo1} cannot be easily derived in this setting and thus more work is needed to perform a fair comparison between the smoothed and the nonparametric bootstrap while implementing post-selection adjustments.

The construction of the CD-plot relies on simulation under both the hypothesized model $G$ and the estimated smooth model $\widehat{F}_m$. The bidirectional acceptance sampling introduced in Section \ref{double}  further extends the LP-smoothed bootstrap to situations where samples from $G$ cannot be easily obtained. In addition to its simple implementation, the bidirectional acceptance sampling allows simultaneous sampling from both $G$ and $\widehat{F}_m$ introducing a substantial computational gain.

All the tools presented in the first five sections of the manuscript are combined in Algorithm 3. Specifically, in this Algorithm we also implement extensions to situations where the postulated model depends on nuisance parameters or a model selection procedure is performed, as described in Section \ref{extensions}.
The latter aspect, in particular, poses a substantial challenge in deriving the asymptotic distribution of the estimator $\widehat{d}_m$. Hence the usefulness and the need for efficient simulation procedures.

Finally, to illustrate the applicability of the methods proposed in practical settings, we analyze COVID-19 hospitalization data. We show that our  approach correctly recovers the underlying distribution of the time (in days) from onset of symptoms to hospitalization of COVID-19 patients. Ultimately, the latter is shown to be a ZINB.

\red{Despite this article mainly focuses on the univariate case, extensions to the multivariate setting are discussed in \citet{iGOF}. Alternative solutions valid in the continuous data regime can be found in \citet{jitkrittum}, whereas \citet{kim} proposes a two-samples test applicable also to multivariate mixed-data distributions. More work is needed, however, to extend the tools presented here and those in \citet{iGOF} to the context of regression. }

\red{
\section*{Supplementary Material}{
\begin{description}
\item[File \texttt{\emph{Supp\_AlgeriZhang.pdf}:}] collects the technical proofs of Theorems \ref{theo1} and \ref{theo2}. 
\item[Folder \texttt{\emph{Codes\_AlgeriZhang}:}] collects the \texttt{R} files and the data needed to reproduce the results discussed in this article. The file \texttt{\emph{Codes\_and\_data\_description.pdf}} available in this folder provides a detailed description of each individual   file.
\end{description}
 }}

{\fontsize{3.5mm}{3.5mm}\selectfont{
\bibliography{biblioLP2}
}}
\end{document}